\documentclass{tlp}
\usepackage{aopmath}

\newtheorem{definition}{Definition} % [section]
\newtheorem{example}{Example} % [section]

\begin{document}
\bibliographystyle{acmtrans}

\long\def\comment#1{}

\title[Fuzzy Linguistic Logic Programming]
			{Fuzzy Linguistic Logic Programming\\
			 and its Applications}

\author[V. H. Le, F. Liu and D. K. Tran]
{VAN HUNG LE, FEI LIU\\
Department of Computer Science and Computer Engineering\\
La Trobe University, Bundoora, VIC 3086, Australia\\
\email{vh2le@students.latrobe.edu.au; f.liu@latrobe.edu.au}
\and DINH KHANG TRAN\\
Faculty of Information Technology\\
Hanoi University of Technology, Vietnam\\
\email{khangtd@it-hut.edu.vn}
}

\pagerange{\pageref{firstpage}--\pageref{lastpage}}
\volume{\textbf{10} (3):}
\jdate{March 2008}
\setcounter{page}{1}
\pubyear{2008}

\maketitle

\label{firstpage}

\begin{abstract}
The paper introduces fuzzy linguistic logic programming, which is a combination of fuzzy logic programming, introduced by P. Vojt\'{a}\v{s}, and hedge algebras in order to facilitate the representation and reasoning on human knowledge expressed in natural languages. In fuzzy linguistic logic programming, truth values are linguistic ones, e.g., \emph{VeryTrue}, \emph{VeryProbablyTrue}, and \emph{LittleFalse}, taken from a hedge algebra of a linguistic truth variable, and linguistic hedges (modifiers) can be used as unary connectives in formulae. 
This is motivated by the fact that humans reason mostly in terms of linguistic terms rather than in terms of numbers, and linguistic hedges are often used in natural languages to express different levels of emphasis.
The paper presents: 
$(i)$ the language of fuzzy linguistic logic programming;
$(ii)$ a declarative semantics in terms of Herbrand interpretations and models; 
$(iii)$ a procedural semantics which directly manipulates linguistic terms to compute a lower bound to the truth value of a query, and proves its soundness;
$(iv)$ a fixpoint semantics of logic programs, and based on it, proves the completeness of the procedural semantics;
$(v)$ several applications of fuzzy linguistic logic programming; and
$(vi)$ an idea of implementing a system to execute fuzzy linguistic logic programs.
\end{abstract}

\section{Introduction}
People usually use words (in natural languages), which are inherently imprecise, vague and qualitative in nature, to describe real world information, to analyse, to reason, and to make decisions. Moreover, in natural languages, linguistic hedges are very often used to state different levels of emphasis. Therefore, it is necessary to investigate logical systems that can directly work with words, and make use of linguistic hedges since such systems will make it easier to represent and reason on knowledge expressed in natural languages.

Fuzzy logic, which is derived from fuzzy set theory, introduced by L. Zadeh, deals with reasoning that is approximate rather than exact, as in classical predicate logic. In fuzzy logic, the truth value domain is not the classical set $\{False,True\}$ or $\{0, 1\}$, but a set of linguistic truth values \cite{Zadeh75b} or the whole unit interval [0,1]. Moreover, in fuzzy logic, linguistic hedges play an essential role in the generation of the values of a linguistic variable and in the modification of fuzzy predicates \cite{Zadeh89}. Fuzzy logic provides us with a very powerful tool for handling imprecision and uncertainty, which are very often encountered in real world information, and a capacity for representing and reasoning on knowledge expressed in linguistic forms.

Fuzzy logic programming, introduced in \citeN{Vo01}, is a formal model of an extension of logic programming without negation working with a truth functional fuzzy logic in narrow sense. In fuzzy logic programming, atoms and rules, which are many-valued implications, are graded to a certain degree in the interval [0,1]. 
Fuzzy logic programming allows a wide variety of many-valued connectives in order to cover a great variety of applications.
A sound and complete procedural semantics is provided to compute a lower bound to the truth value of a query.
Nevertheless, no proofs of extended versions of Mgu and Lifting lemmas are given.
Fuzzy logic programming has applications such as threshold computation, a data model for flexible querying \cite{PV01}, and fuzzy control \cite{Gerla05}. 

The theory of hedge algebras, introduced in \citeANP{Ho90} \citeNN{Ho90,HoWech92}, forms an algebraic approach to a natural qualitative semantics of linguistic terms in a term domain. The hedge-algebra-based semantics of linguistic terms is qualitative, relative, and dependent on the order-based structure of the term domain. 
Hedge algebras have been shown to have a rich algebraic structure to represent linguistic domains \cite{HoKhang99}, and
the theory can be effectively applied to problems such as linguistic reasoning \cite{HoKhang99} and fuzzy control \cite{Ho08}. 
The notion of an \emph{inverse mapping of a hedge} is defined in \citeN{DK06} for monotonic hedge algebras, a subclass of linear hedge algebras.

In this work, we integrate fuzzy logic programming and hedge algebras to build a logical system that facilitates the representation and reasoning on knowledge expressed in natural languages. 
In our logical system, the set of truth values is that of linguistic ones taken from a hedge algebra of a linguistic truth variable. Furthermore, we consider only finitely many truth values. On the one hand, this is due to the fact that normally, people use finitely many degrees of quality or quantity to describe real world applications which are granulated \cite{Zadeh97}. On the other hand, it is reasonable to provide a logical system suitable for computer implementation. In fact, the finiteness of the truth domain allows us to obtain the Least Herbrand model for a finite logic program after a finite number of iterations of an immediate consequences operator. 
Moreover, we allow the use of linguistic hedges as unary connectives in formulae to express different levels of accentuation on fuzzy predicates. The procedural semantics in \citeN{Vo01} is extended to deduce a lower bound to the truth value of a query by directly computing with linguistic terms.

The paper is organised as follows: 
the next section gives a motivating example for the development of fuzzy linguistic logic programming; Section 3 presents linguistic truth domains taken from hedge algebras of a truth variable, inverse mappings of hedges, many-valued modus ponens w.r.t. such domains; Section 4 presents the theory of fuzzy linguistic logic programming, defining the language, declarative semantics, procedural semantics, and fixpoint semantics, and proving the soundness and completeness of the procedural semantics; Section 5 and Section 6 respectively discuss several applications and an idea for implementing a system where such logic programs can be executed; the last section summarises the paper.

\section{Motivation}
Our motivating example is adapted from the hotel reservation system described in \citeN{NOHW95}. Here, we use logic programming notation. A rule to find a convenient hotel for a business trip can be defined as follows:
\begin{eqnarray*}
convenient\_hotel(Business\_location,Time, Hotel)\leftarrow \\
\wedge(near\_to(Business\_location,Hotel), \\
reasonable\_cost(Hotel,Time), \\
fine\_building(Hotel)). \mbox{with truth value=}VeryTrue 
\end{eqnarray*}
That is, a hotel is regarded to be convenient for a business trip if it is near the business location, has a reasonable cost at the considered time, and is a fine building. 

Here, \emph{fine\_building(Hotel)} is an \emph{atomic formula} (atom), which is a fuzzy predicate symbol with a list of arguments, having a truth value. There is an option that the truth value of \emph{fine\_building} of a hotel is a number in [0,1] and is calculated by a function of its age as in \citeN{NOHW95}. However, in fact, the age of a hotel may not be enough to reflect its fineness since the fineness also depends on the construction quality and the surroundings. Similarly, the truth value of \emph{reasonable\_cost} can be computed as a function of the hotel rate at the time. Nevertheless, since the rate varies from season to season, the function should be modified accordingly to reflect the reasonableness for a particular time. Thus, a more realistic and appropriate way is to assess the fineness and the reasonableness of the cost of a hotel using linguistic truth values, e.g., \emph{ProbablyTrue}, after considering all possible factors. 

Note that there can be more than one way to define the convenience of a hotel, and the above rule is only one of them. Furthermore, since any of such rules may not be absolutely true for everybody, each rule should have a degree of truth (truth value). For example, \emph{VeryTrue} is the truth value of the above rule.

In addition, since linguistic hedges are usually used to state different levels of emphasis, we desire to use them to express different degrees of requirements on the criteria. For example, if we want to emphasise closeness, we can use the formula \emph{\textbf{Very} near\_to(Business\_location,Hotel)} instead of \emph{near\_to(Business\_location,Hotel)} in the rule, and if we do not care much about the cost, we can relax the criterion by using the hedge \emph{Probably} for the atom \emph{reasonable\_cost(Hotel,Time)}. Thus, the rule becomes:
\begin{eqnarray*}
convenient\_hotel(Business\_location,Time, Hotel)\leftarrow \\
\wedge(Very~near\_to(Business\_location,Hotel), \\
Probably~reasonable\_cost(Hotel,Time), \\
fine\_building(Hotel)). \mbox{with truth value=}VeryTrue 
\end{eqnarray*}
In our opinion, in order to model knowledge expressed in natural languages, a formalism should address the twofold usage of linguistic hedges, i.e., in generating linguistic values and in modifying predicates. 
To the best of our knowledge, no existing frameworks of logic programming have addressed the problem of using linguistic truth values as well as allowing linguistic hedges to modify fuzzy predicates.

\section{Hedge algebras and linguistic truth domains}
\subsection{Hedge algebras}
Since the mathematical structures of a given set of truth values play an important role in studying the corresponding logics, we present here an appropriate mathematical structure of a linguistic domain of a linguistic  variable \emph{Truth} in particular, and that of any linguistic variable in general.

In an algebraic approach, values of the linguistic variable \emph{Truth} such as \emph{True}, \emph{VeryTrue}, \emph{ProbablyFalse}, \emph{VeryProbablyFalse}, and so on can be considered to be generated from a set of generators (primary terms) $G = \{False, True\}$ using hedges from a set $H=\{Very,More,Probably,$ $...\}$ as unary operations. There exists a natural ordering among these values, with $a\leq b$ meaning that $a$ indicates a degree of truth less than or equal to $b$. For example, $True < VeryTrue$ and $False <LittleFalse$, where $a<b$ iff $a\leq b$ and $a\neq b$. The relation $\leq$ is called the \emph{semantically ordering relation} (SOR) on the term domain, denoted by $X$. 

There are natural semantic properties of linguistic terms and hedges that can be formulated in terms of the SOR as follows. Let \emph{V, M, L, P}, and \emph{A} stand for the hedges \emph{Very, More, Little, Probably}, and \emph{Approximately}, respectively.

$(i)$ Hedges either increase or decrease the meaning of terms they modify, so they can be regarded as \emph{ordering operations}, i.e., $\forall h\in H, \forall x\in X, \mbox{ either } hx\geq x \mbox{ or } hx\leq x$. 
The fact that a hedge $h$ modifies terms more than or equal to another hedge $k$, i.e., $\forall x \in X$, $hx\leq kx\leq x$ or $x\leq kx\leq hx$, is denoted by $h\geq k$. Note that since the sets $H$ and $X$ are disjoint, we can use the same notation $\leq$ for different ordering relations on $H$ and on $X$ without any confusion.
For example, we have $L>P$ ($h>k$ iff $h\geq k$ and $h\neq k$) since, for instance, $LTrue<PTrue<True$ and $LFalse>PFalse>False$.

$(ii)$ A hedge has a semantic effect on others, i.e., it either strengthens or weakens the degree of modification of other hedges. If $h$ strengthens the degree of modification of $k$, i.e., $\forall x \in X$, $hkx\leq kx\leq x$ or $x\leq kx \leq hkx$, then it is said that $h$ is \emph{positive} w.r.t. $k$; if $h$ weakens the degree of modification of $k$, i.e., $\forall x \in X$, $kx\leq hkx\leq x$ or $x\leq hkx\leq kx$, then it is said that $h$ is \emph{negative} w.r.t. $k$. For instance, $V$ is positive w.r.t. $M$ since, e.g., $VMTrue>MTrue>True$; $V$ is negative w.r.t. $P$ since, e.g., $PTrue<VPTrue<True$.

$(iii)$ An important semantic property of hedges, called \emph{semantic heredity}, is that hedges change the meaning of a term a little, but somewhat preserve the original meaning. Thus, if there are two terms $hx$ and $kx$, where $x\in X$, such that $hx\leq kx$, then all terms generated from $hx$ using hedges are less than or equal to all terms generated from $kx$. This property is formulated by: $(a)$ If $hx\leq kx$, then $H(hx)\leq H(kx)$, where $H(u)$ denotes the set of all terms generated from $u$ by means of hedges, i.e., $H(u) = \{\sigma u|\sigma\in H^{*}\}$, where $H^{*}$ is the set of all strings of symbols in $H$ including the empty one. For example, since $MTrue\leq VTrue$, we have $VMTrue\leq LVTrue$ and $H(MTrue)\leq H(VTrue)$; $(b)$ If two terms $u$ and $v$ are incomparable, then all terms generated from $u$ are incomparable to all terms generated from $v$. For example, since $AFalse$ and $PFalse$ are incomparable, $VAFalse$ and $MPFalse$ are incomparable too. 

Two terms $u$ and $v$ are said to be \emph{independent} if $u\notin H(v)$ and $v\notin H(u)$. For example, $VTrue$ and $PMTrue$ are independent, but $VTrue$ and $LVTrue$ are not since $LVTrue\in H(VTrue)$.

\begin{definition}[Hedge algebra] \cite{Ho90}
An abstract algebra $\underline{X}=(X,G,H,\leq)$, where $X$ is a term domain, $G$ is a set of primary terms, $H$ is a set of linguistic hedges, and $\leq$ is an SOR on $X$, is called a \emph{hedge algebra} (HA) if it satisfies the following:

(A1) Each hedge is either positive or negative w.r.t. the others, including itself;

(A2) If terms $u$ and $v$ are independent, then, for all $x\in H(u)$, we have $x\notin H(v)$. In addition, if $u$ and $v$ are incomparable, i.e., $u \not< v$ and $v \not< u$, then so are $x$ and $y$, for every $x\in H(u)$ and $y\in H(v)$;

(A3) If $x \neq hx$, then $x\notin H(hx)$, and if $h\neq k$ and $hx \leq kx$, then $h'hx\leq k'kx$, for all $h, k, h', k' \in H$ and $x \in X$. Moreover, if $hx\neq kx$, then $hx$ and $kx$ are independent;

(A4) If $u \notin H(v)$ and $u\leq v$ ($u\geq v$), then $u \leq hv$ ($u\geq hv$) for any $h\in H$.
\end{definition}
Axioms (A2)-(A4) are a weak formulation of the semantic heredity of hedges. 

Given a term $u$ in $X$, the expression $h_{n}...h_{1}u$ is called a \emph{representation} of $x$ w.r.t. $u$ if $x = h_{n}...h_{1}u$, and, furthermore, it is called a \emph{canonical representation} of $x$ w.r.t. $u$ if $h_{n} h_{n-1}...h_{1}u \neq  h_{n-1}...h_{1}u$. 

The following proposition shows how to compare any two terms in $X$. The notation $x_{u|j}$ denotes the suffix of length $j$ of a representation of $x$ w.r.t. $u$, i.e., for $x = h_{n}...h_{1}u$, $x_{u|j}$ = $h_{j-1}...h_{1}u$, where $2\leq j\leq n+1$, and $x_{u|1} = u$. Let $I\notin H$ be an artificial hedge called the \emph{identity} on $X$ defined by the rule $\forall x\in X$, $Ix = x$.

\begin{proposition}
\label{prop1}
\cite{HoWech92}
Let $x=h_{n}...h_{1}u$, $y=k_{m}...k_{1}u$ be two canonical representations of $x$ and $y$ w.r.t. $u$, respectively. Then, there exists the largest $j\leq min(m,n)+1$ 
(here, as a convention it should be understood that if  $j=min(m,n)+1$, then $h_{j}=I$, for $j=n+1$, and $k_{j}=I$, for $j=m+1$)
such that $\forall i < j$, $h_i=k_i$, and 

($i$) $x=y$ iff $n=m$ and $h_{j}x_{u|j}=k_{j}x_{u|j}$;

($ii$) $x<y$ iff $h_jx_{u|j}<k_jx_{u|j}$; 

($iii$) $x$ and $y$ are incomparable iff $h_jx_{u|j}$ and $k_jx_{u|j}$ are incomparable. 
\end{proposition}

\subsection{Linear symmetric hedge algebras}
Since we allow hedges to be unary connectives in formulae, there is a need to be able to compute the truth value of a hedge-modified formula from that of the original. To this end, the notion of \emph{an inverse mapping of a hedge} is utilised. In order to define this notion, we restrict ourselves to linear HAs. 

The set of primary terms $G$ usually consists of two comparable ones, denoted by $c^{-}<c^{+}$. For the variable \emph{Truth}, we have $c^{-}=False<c^{+}=True$. Such HAs are called \emph{symmetric} ones.  
For symmetric HAs, the set of hedges $H$ can be divided into two disjoint subsets $H^{+}$ and $H^{-}$ defined as $H^{+} = \{h| hc^{+}>c^{+}\}$ and $H^{-} = \{h| hc^{+}<c^{+}\}$. 
Two hedges $h$ and $k$ are said to be \emph{converse} if $\forall x \in X$, $hx\leq x$ iff $kx\geq x$, i.e., they are in different subsets; $h$ and $k$ are said to be \emph{compatible} if $\forall x \in X$, $hx\leq x$ iff $kx\leq x$, i.e., they are in the same subset. 

Two hedges in each of sets $H^{+}$ and $H^{-}$ may be comparable, e.g., $L$ and $P$, or incomparable, e.g., $A$ and $P$. Thus, $H^{+}$ and $H^{-}$ become posets. 

\begin{definition}[Linear symmetric hedge algebra]
A symmetric HA $\underline{X} = (X,G=\{c^{-},c^{+}\},H,\leq)$ is said to be a \emph{linear symmetric} HA (lin-HA, for short) if the set of hedges $H$ is divided into $H^{+} = \{h| hc^{+}>c^{+}\}$ and $H^{-} = \{h| hc^{+}<c^{+}\}$, and $H^{+}$ and $H^{-}$ are linearly ordered. 
\end{definition}
\begin{example} \label{ex1}
Consider an HA $\underline{X} = (X,G=\{c^{-},c^{+}\},H=\{V, M, P,L\},\leq)$. $\underline{X}$ is a lin-HA as follows. $V$ and $M$ are positive w.r.t. $V$, $M$, and $L$, and negative w.r.t. $P$; $P$ is positive w.r.t. $P$, and negative w.r.t. $V$, $M$, and $L$; $L$ is positive w.r.t. $P$, and negative w.r.t. $V$, $M$, and $L$. $H$ is decomposed into $H^{+}=\{V,M\}$ and $H^{-}=\{P,L\}$. Moreover, in $H^{+}$, we have $M<V$, and in $H^{-}$, we have $P<L$. 
\end{example}

\begin{definition}[Sign function]
\cite{Ho90}
\label{def3}
A function $Sign: X\rightarrow \{-1,0,+1\}$ is a mapping defined recursively as follows, where $h, h'\in H$ and $c\in \{c^{-},c^{+}\}$:

a) $Sign(c^{-}) = -1$, $Sign(c^{+}) = +1$; 

b) $Sign(hc) = -Sign(c)$ if either $h\in H^{+}$ and $c=c^{-}$ or $h\in H^{-}$ and $c=c^{+}$; 

c) $Sign(hc) = Sign(c)$ if either $h\in H^{+}$ and $c=c^{+}$ or $h\in H^{-}$ and $c=c^{-}$; 

d) $Sign(h'hx) = -Sign(hx)$, if  $h'hx\neq hx$, and  $h'$ is negative w.r.t. $h$;

e) $Sign(h'hx) = Sign(hx)$, if $h'hx\neq hx$, and  $h'$ is positive w.r.t. $h$;

f) $Sign(h'hx) = 0$ if $h'hx = hx$. 
\end{definition}
Based on the function $Sign$, we have a criterion to compare $hx$ and $x$ as follows:

\begin{proposition}
\label{prop99}
\cite{Ho90}
For any $h$ and $x$, if $Sign(hx)= +1$, then $hx > x$, and if $Sign(hx) = -1$, then  $hx < x$.
\end{proposition}
In \citeN{HoWech92}, HAs are extended by augmenting two artificial hedges $\Phi$ and $\Sigma$ defined as $\Phi(x) = infimum(H(x))$ and $\Sigma(x) = supremum(H(x))$, for all $x\in X$. 
An HA is said to be \emph{free} if $\forall x\in X$ and $\forall h\in H$, $hx\neq x$.
It is shown that, for a free lin-HA of the variable \emph{Truth} with $H \neq\emptyset$, $\Phi(c^{+}) = \Sigma(c^{-})$, $\Sigma(c^{+}) = 1$ (\emph{AbsolutelyTrue}), and $\Phi(c^{-}) = 0$ (\emph{AbsolutelyFalse}). Let us put $W=\Phi(c^{+}) = \Sigma(c^{-})$ (called the \emph{middle truth value}); we have $0<c^{-}<W<c^{+}<1$.

\begin{definition} [Linguistic truth domain]
A \emph{linguistic truth domain} $\overline{X}$ taken from a lin-HA $\underline{X} = (X,\{c^{-},c^{+}\}, H,\leq)$ is defined as $\overline{X}=X\cup \{0,W,1\}$, where $0, W$, and $1$ are the least, the neutral, and the greatest elements of $\overline{X}$, respectively.
\end{definition}
\begin{proposition}\cite{HoWech92}
\label{prop2}
For any lin-HA $\underline{X} = (X,G,H,\leq)$, the linguistic truth domain $\overline{X}$ is linearly ordered.
\end{proposition}
The usual operations are defined on $\overline{X}$ as follows:
$(i)$ \emph{negation}: given $x = \sigma c$, where $\sigma\in H^{*}$ and $c\in\{c^{+},c^{-}\}$, $y$ is called the \emph{negation} of $x$, denoted by $y = -x$, if $y = \sigma c'$ and $\{c, c'\}$ = $\{c^{+}, c^{-}\}$. For example, $hc^{+}$ is the negation of $hc^{-}$. In particular, $-1=0$, $-0=1$, and $-W=W$;
$(ii)$ \emph{conjunction}: $x\wedge y$ = min($x, y$);
$(iii)$ \emph{disjunction}: $x\vee y$ = max($x, y$). 

\begin{proposition}\cite{HoWech92}
\label{prop3}
For any lin-HA $\underline{X} = (X,G,H,\leq)$, the following hold:
$(i)$ $-hx = h(-x)$ for any $h\in H$;
$(ii)$ $--x=x$;
$(iii)$ $x < y$ iff $-x>-y$.
\end{proposition}
It is shown that the identity hedge $I$ is the least element of the sets $H^{+}\cup \{I\}$ and $H^{-}\cup \{I\}$, i.e., $\forall h \in H$, $h\geq I$. 

\begin{definition}[Extended ordering relation]
An \emph{extended ordering relation} on $H \cup \{I\}$, denoted by $\leq_e$, is defined based on the ordering relations on $H^{+}\cup \{I\}$ and $H^{-}\cup \{I\}$ as follows.
Given $h,k\in H\cup\{I\}$, $h\leq_e k$ iff: 
($i$) $h\in H^{-}, k\in H^{+}$; or 
($ii$) $h,k\in H^{+}\cup\{I\}$ and $h\leq k$; or 
($iii$) $h,k\in H^{-}\cup\{I\}$ and $h\geq k$. 
We denote $h<_e k$ iff $h\leq_e k$ and $h\neq k$.
\end{definition}
\begin{example} \label{ex2}
For the HA in Example \ref{ex1}, in $H\cup\{I\}$ we have $L<_e P<_e I<_e M<_e V$.
\end{example}
It is straightforward to show the following:
\begin{proposition}
\label{prop4}
For all $h,k\in H\cup \{I\}$, if $h<_ek$, then $hc^{+}< kc^{+}$.
\end{proposition}

\subsection{Inverse mappings of hedges}
In fuzzy logic, knowledge is usually represented in terms of pairs consisting of a \emph{vague sentence} and its \emph{degree of truth}, which is also expressed in linguistic terms. 
A vague sentence can be represented by an expression $u(x)$, where $x$ is a variable or a constant, and $u$ is a fuzzy predicate. For example, the assertion ``\emph{It is quite true that John is studying hard}" can be represented by a pair $(study\_hard(john),Quite True)$. 

According to \citeANP{Zadeh79} \citeNN{Zadeh79,Zadeh75}, the following assessments can be considered to be approximately semantically equivalent: ``\emph{It is very true that Lucia is young}" and ``\emph{It is true that Lucia is very young}". That means if we have $(young(lucia),VeryTrue)$, we also have $(Very\:young(lucia),True)$. Thus, the hedge ``\emph{Very}" can be moved from the truth value to the fuzzy predicate. This is generalised to the following rule:
\[(R1)\;\;\;(u(x),hTrue)\Rightarrow(hu(x),True)\]
However, the rule is not complete, i.e., in some cases we cannot use it to deduce the truth value of a hedge-modified fuzzy predicate from that of the original. 
For instance, given $(young(lucia),VeryTrue)$, we cannot compute the truth value of $Probably\:young(lucia)$ using the above rule. The notion of an \emph{inverse mapping of a hedge}, which is an extension of Rule $(R1)$, provides a solution to this problem. 

The idea behind this notion is that the truth value of a hedge-modified fuzzy predicate can be a function of that of the original. 
In other words, if we modify a fuzzy predicate by a hedge, its truth value will be changed by the inverse mapping of that hedge. 
Now, we will work out the conditions that an inverse mapping of a hedge should satisfy. We denote the inverse mapping of a hedge $h$ by $h^{-}$. First, since $h^{-}$ is an extension of Rule $(R1)$, we should have $h^{-}(hTrue) = True$. Second, intuitively, the more true a fuzzy predicate is, the more true is its hedge-modified one, so $h^{-}$ should be monotone, i.e., if $x\geq y$, then $h^{-}(x)\geq h^{-}(y)$. 

Third, it seems to be natural that by modifying a fuzzy predicate using a hedge in $H^{+}$ such as \emph{Very} or \emph{More}, we accentuate the fuzzy predicate, so the truth value should decrease. For example, the truth value of 
$Very\:young(lucia)$ should be less than that of $young(lucia)$. Similarly, by applying a hedge in $H^{-}$ such as \emph{Probably} or \emph{Little}, we deaccentuate the fuzzy predicate; thus, the truth value should increase. For example, the truth value of $Probably\:high\_income(tom)$ should be greater than that of $high\_income(tom)$. This is also in accordance with the fuzzy-set-based interpretation of hedges \cite{Zadeh72}, in which hedges such as \emph{Very} are called \emph{accentuators} and can be defined as $Very\;x = x^{1+\alpha}$, where $x$ is a fuzzy predicate expressed by a \emph{fuzzy set} and $\alpha>0$, and hedges such as \emph{Probably} are called \emph{deaccentuators} and can be defined as $Probably\;x = x^{1-\alpha}$ (note that the degree of membership of each element in $x$ is in [0,1]). In summary, this can be formulated as: for all $h,k \in H\cup \{I\}$ such that $h\leq_e k$ and for all $x$, we should have $h^{-}(x)\geq k^{-}(x)$. As a convention, we always assume that for all $x$, $I^{-}(x)=x$.

\begin{definition}[An inverse mapping of a hedge]
Given a lin-HA $\underline{X}=(X,\{c^{+}, c^{-}\},H,\leq)$ and a hedge $h\in H$, a mapping $h^{-}: \overline{X} \rightarrow \overline{X}$ is called an \emph{inverse mapping} of $h$ iff it satisfies the following conditions:
\begin{eqnarray}
h^{-}(h c^{+}) =  c^{+} \label{eq:im1} \\
x \geq y \Rightarrow h^{-}(x) \geq h^{-}(y) \label{eq:im2}\\
h\leq_e k \Rightarrow h^{-}(x)\geq k^{-}(x) \label{eq:im3}
\end{eqnarray}
where $k^{-}$ is an inverse mapping of another hedge $k\in H\cup \{I\}$.
\end{definition}
Since 0, W, and 1 are fixed points, i.e., $\forall x\in \{0,W,1\}$ and $\forall h\in H$, $hx=x$ \cite{HoWech92}, it is reasonable to assume that $\forall h\in H$, $h^{-}(0)=0, h^{-}(W)=W$, and $h^{-}(1)=1$. 

We show why we have to use lin-HAs in order to define the notion of an inverse mapping of a hedge. Consider an HA containing two incomparable hedges $P\;(Probably), A\;(Approximately)\in H^{-}$. We can see that since $Ac^{+}$ and $Pc^{+}$ are incomparable, $P^{-}(Ac^{+})$ and $P^{-}(Pc^{+})=c^{+}$ should be either incomparable or equal. The two values cannot be incomparable since every truth value is comparable to $c^{+}$ and $c^{-}$, and it might not be very meaningful to keep both $P$ and $A$ in the set of hedges if we have $P^{-}(Ac^{+})=P^{-}(Pc^{+})=c^{+}$.

Inverse mappings of hedges always exist; in the following, we give an example of inverse mappings of hedges for a general lin-HA. 

\begin{example} \label{ex3}
Consider a lin-HA $\underline{X}=(X,\{c^{+}, c^{-}\},H,\leq)$ with $H^{-}=\{h_{-q}, h_{-q+1}, ..., h_{-1}\}$ and $H^{+}=\{h_{1}, h_{2}, ..., h_{p}\}$, where $p,q\geq 1$. Let us denote $h_{0}=I$. Without loss of generality, we suppose that $h_{-q}> h_{-q+1}> ...> h_{-1}$ and $h_{1}< h_{2}< ...< h_{p}$. Therefore, we have $h_{-q}<_e h_{-q+1}<_e ...<_e h_{-1}<_e h_{0}<_e h_{1}<_e h_{2}$ $<_e ...<_e h_{p}$, and thus $h_{-q}c^{+}< ...< h_{-1}c^{+}< c^{+}< h_{1}c^{+}< ...< h_{p}c^{+}$. We always assume that, for all $k_1,k_2\in H$ and $c\in \{c^{+}, c^{-}\}$, $k_2k_1c\neq k_1c$, i.e., $Sign(k_2k_1c) \neq 0$.

First, we build inverse mappings of hedges $h_{r}^{-}(x)$, for all $x\in H(c^{+})$, as follows:

$(i)$ $x=c^{+}$. For all $r$ such that $-min(p,q)\leq r\leq min(p,q)$, we put $h_{r}^{-}(c^{+})=h_{-r}c^{+}$. In particular, $h_{0}^{-}(c^{+})=h_{0}c^{+}=c^{+}$. If $p>q$, for all $q+1\leq r\leq p$, $h_{r}^{-}(c^{+})=W$. If $p<q$, for all $-(p+1)\geq r\geq -q$, $h_{r}^{-}(c^{+})=1$. It can be easily verified that, for all $h\in H\cup \{I\}$, $h^{-}(c^{+})$ satisfies Condition (\ref{eq:im3}).

$(ii)$ $x=\sigma h_{s} c^{+}$, where $\sigma \in H^{*}$ and $h_{s}\neq I$, i.e., $s\neq 0$. If $r=s$, we put $h_{r}^{-}(\sigma h_{r} c^{+})= c^{+}$; hence, Condition (\ref{eq:im1}) is satisfied. Otherwise, we have $r\neq s$. If $s-r< -q$, we put $h_{r}^{-}(\sigma h_{s}c^{+})=W$; if $s-r> p$, we put $h_{r}^{-}(\sigma h_{s}c^{+})=1$. Otherwise, we have $-q\leq s-r\leq p$.

For a certain hedge $k$, $Sign(h_pkc^{+})$ can be either -1 or +1 . If $Sign(h_pkc^{+})=+1$, by Proposition \ref{prop99}, we have $kc^{+}< h_pkc^{+}$. Thus, it follows that $h_{-q}kc^{+}< ... < h_{-1}kc^{+}< kc^{+}<h_{1}kc^{+}< ... < h_{p}kc^{+}$. For example, we have $Sign(VPc^{+})=+1$ and $LPc^{+}< PPc^{+}< Pc^{+}< MPc^{+}< VPc^{+}$. Similarly, if $Sign(h_pkc^{+})=-1$, we have $h_{-q}kc^{+}> ... > h_{-1}kc^{+}> kc^{+}>h_{1}kc^{+}> ... > h_{p}kc^{+}$. For instance, we have $Sign(VLc^{+})=-1$ and $LLc^{+}> PLc^{+}> Lc^{+}> MLc^{+}> VLc^{+}$. In summary, the ordering of the elements in the set $\{h_{t}kc^{+}:-q\leq t\leq p\}$ can have one of the two above reverse directions.
Therefore, for a pair $(s,s-r)$, there are two cases:

(a) The orderings of the elements in the sets $\{h_{t}h_{s}c^{+}:-q\leq t\leq p\}$  and $\{h_{t}h_{s-r}c^{+}:-q\leq t\leq p\}$
have the same direction, i.e., we have $h_{-q}h_{s}c^{+}< ... < h_{-1}h_{s}c^{+}< h_{s}c^{+}< h_{1}h_{s}c^{+}< ... < h_{p}h_{s}c^{+}$ and $h_{-q}h_{s-r}c^{+}< ... < h_{-1}h_{s-r}c^{+}< h_{s-r}c^{+}< h_{1}h_{s-r}c^{+}< ... < h_{p}h_{s-r}c^{+}$, or $h_{-q}h_{s}c^{+}> ... > h_{-1}h_{s}c^{+}> h_{s}c^{+}> h_{1}h_{s}c^{+}> ... > h_{p}h_{s}c^{+}$ and $h_{-q}h_{s-r}c^{+}> ... > h_{-1}h_{s-r}c^{+}> h_{s-r}c^{+}> h_{1}h_{s-r}c^{+}> ... > h_{p}h_{s-r}c^{+}$. In this case, we put $h_{r}^{-}(\sigma h_{s}c^{+})=\sigma h_{s-r} c^{+}$. 

(b) The orderings have reverse directions, i.e., we have $h_{-q}h_{s}c^{+}< ... < h_{-1}h_{s}c^{+}< h_{s}c^{+}< h_{1}h_{s}c^{+}< ... < h_{p}h_{s}c^{+}$ and $h_{-q}h_{s-r}c^{+}> ... > h_{-1}h_{s-r}c^{+}> h_{s-r}c^{+}> h_{1}h_{s-r}c^{+}> ... > h_{p}h_{s-r}c^{+}$, or 
$h_{-q}h_{s}c^{+}> ... > h_{-1}h_{s}c^{+}> h_{s}c^{+}> h_{1}h_{s}c^{+}> ... > h_{p}h_{s}c^{+}$ and $h_{-q}h_{s-r}c^{+}< ... < h_{-1}h_{s-r}c^{+}< h_{s-r}c^{+}< h_{1}h_{s-r}c^{+}< ... < h_{p}h_{s-r}c^{+}$.
We put $h_{r}^{-}(\sigma h_{s}c^{+})=\delta h_{s-r}c^{+}$, where $\delta$ is obtained as follows. 
If $\sigma$ is empty, then so is $\delta$. Otherwise, suppose that $\sigma=\sigma'h_{t}$, where $t\neq 0$. 
If $-q\leq -t \leq p$, we put $\delta = h_{-t}$; if $-t<-q$, then $\delta = h_{-q}$; if $-t>p$, then $\delta = h_{p}$.  It can be seen that what we have done here is to make inverse mappings of hedges monotone.

In particular, if $r=0$, then $s=s-r$. Thus, (b) is not the case, and by (a), we have
$h_{0}^{-}(\sigma h_{s} c^{+})=\sigma h_{s} c^{+}$; this complies with the assumption $I^{-}(x)=x$, for all $x$.

Second, for $x\in H(c^{-})$, we define $h^{-}_{r}(x)$ based on the above case as follows. Note that from $x\in H(c^{-})$, we have $-x\in H(c^{+}$). 
If $-min(p,q)\leq r\leq min(p,q)$, we put $h_{r}^{-}(x)=-h_{-r}^{-}(-x)$; if $p>q$, for all $q+1\leq r\leq p$, $h_{r}^{-}(x)=-h_{-q}^{-}(-x)$; if $p<q$, for all $-(p+1)\geq r\geq -q$, $h_{r}^{-}(x)=-h_{p}^{-}(-x)$. 

Finally, as usual, $h^{-}(1)=1$, $h^{-}(W)=W$, and $h^{-}(0)=0$, for all $h$.

It can be easily seen that, for all $x\in H(c^{+})$ and $h\in H\cup \{I\}$, $h^{-}(x)\in H(c^{+})\cup \{W,1\}$, and, 
for all $x\in H(c^{-})$ and $h\in H\cup \{I\}$, $h^{-}(x)\in H(c^{-})\cup \{W,0\}$.
\end{example}
It has been shown in the above example that the inverse mappings satisfy Condition (\ref{eq:im1}). In the following, we prove that they also satisfy Conditions (\ref{eq:im2}) and (\ref{eq:im3}).
\begin{proposition}
The mappings defined above satisfy Condition (\ref{eq:im3}), i.e., $h\leq_e k \Rightarrow h^{-}(x)\geq k^{-}(x)$.
\end{proposition}
\begin{proof}
We prove that if $h <_e k$, then $h^{-}(x)\geq k^{-}(x)$.
Assume that $h = h_{r_{1}}$, $k = h_{r_{2}}$, where $r_{1}<r_{2}$.

First, we prove the case $x\in H(c^{+})$. 
The case $x=c^{+}$ has been shown to satisfy Condition (\ref{eq:im3}) in Example \ref{ex3}.
Consider the case $x=\sigma h_{s}c^{+}$, where $s\neq 0$. 
From $r_{1}<r_{2}$ we have $s-r_{1}>s-r_{2}$. 
The case $s-r_{2}<-q$, i.e., $h_{r_{2}}^{-}(\sigma h_{s}c^{+})=W$, is trivial; so is the case $s-r_{1}>p$, i.e., $h_{r_{1}}^{-}(\sigma h_{s}c^{+})=1$. Otherwise, $-q\leq s-r_{2}<s-r_{1}\leq p$; thus, $h^{-}(x)=\delta_{1} h_{s-r_{1}}c^{+}$ and $k^{-}(x)=\delta_{2} h_{s-r_{2}}c^{+}$, for some $\delta_{1}$ and $\delta_{2}$.
Since $h_{s-r_{1}}c^{+}> h_{s-r_{2}}c^{+}$, by Proposition \ref{prop1}, we have $h^{-}(x)> k^{-}(x)$.

Second, consider the case $x\in H(c^{-})$. Since $-x\in H(c^{+})$, from the above case, we have, for all $t$, $h_{p}^{-}(-x)\leq h_{t}^{-}(-x)\leq h_{-q}^{-}(-x)$, and by Proposition \ref{prop3}, $-h_{p}^{-}(-x)\geq -h_{t}^{-}(-x)\geq -h_{-q}^{-}(-x)$. 
If $-r_{1}>p$, then $h_{r_{1}}^{-}(x)=-h_{p}^{-}(-x)$; if $-r_{2}<-q$, then $h_{r_{2}}^{-}(x)=-h_{-q}^{-}(-x)$. Thus, we always have $h_{r_{1}}^{-}(x)\geq h_{r_{2}}^{-}(x)$.
Otherwise, $p\geq -r_{1}>-r_{2}\geq -q$; thus, $h_{r_{1}}^{-}(x)=-h_{-r_{1}}^{-}(-x)$ and $h_{r_{2}}^{-}(x)=-h_{-r_{2}}^{-}(-x)$. We have $h_{-r_{1}}^{-}(-x) \leq h_{-r_{2}}^{-}(-x)$; thus, $-h_{-r_{1}}^{-}(-x) \geq -h_{-r_{2}}^{-}(-x)$, i.e., $h_{r_{1}}^{-}(x)\geq h_{r_{2}}^{-}(x)$.

Finally, for $x\in \{0,W,1\}$, we have $h^{-}(x)=k^{-}(x)=x$.
\end{proof}

\begin{proposition}
The mappings defined above satisfy Condition (\ref{eq:im2}), i.e., $x \geq y \Rightarrow h^{-}(x) \geq h^{-}(y)$.
\end{proposition}
\begin{proof}
Suppose $x>y$. Consider $h_{r}^{-}(x)$ and $h_{r}^{-}(y)$, for some $r$. 

First, we prove the case $x,y\in H(c^{+})$.
There are three possible cases:

(1) $x=c^{+}$ and $y=\sigma h_{t}c^{+}$, where $t<0$. 
If $t-r<-q$, then $h_{r}^{-}(y)=W\leq h_{r}^{-}(x)$; if $-r>p$, then $h_{r}^{-}(x)=1\geq h_{r}^{-}(y)$. Otherwise, $-q\leq t-r<-r\leq p$, thus $h_r^{-}(x)=h_{-r}c^{+}$ and $h_r^{-}(y)=\delta h_{t-r}c^{+}$. Since $h_{-r}c^{+}>h_{t-r}c^{+}$, we have $h_{r}^{-}(x)> h_{r}^{-}(y)$.

(2) $y=c^{+}$ and $x=\sigma h_{t}c^{+}$, where $t>0$. The proof is similar to that of (1).

(3) $x=\sigma h_{t}c^{+}$ and $y=\delta h_{s}c^{+}$, where $t\geq s$. If $s-r<-q$, then $h_{r}^{-}(y)=W\leq h_{r}^{-}(x)$, and if $t-r>p$, then $h_{r}^{-}(x)=1\geq h_{r}^{-}(y)$. Otherwise, $-q\leq s-r\leq t-r\leq p$; thus, $h_r^{-}(x)=\sigma' h_{t-r}c^{+}$ and $h_r^{-}(y)=\delta' h_{s-r}c^{+}$. There are two cases:

(3.1) $t-r>s-r$. Since $h_{t-r}c^{+}>h_{s-r}c^{+}$, by Proposition \ref{prop1}, $h_r^{-}(x)> h_r^{-}(y)$. 

(3.2) $t=s$. Suppose $x=\sigma_{1}h_{m}h_{s}c^{+}$ and $y=\delta_{1}h_{n}h_{s}c^{+}$, where if $m = 0$, then $\sigma_{1}$ is empty, and if $n = 0$, then $\delta_{1}$ is empty. There are two cases:

(3.2.1) $m\neq n$. Since $x>y$, by Proposition \ref{prop1}, $h_{m}h_{s}c^{+}>h_{n}h_{s}c^{+}$. 
If $h_{m}h_{s-r}c^{+}>h_{n}h_{s-r}c^{+}$, by (a), $h_r^{-}(x)=\sigma_{1}h_{m}h_{s-r}c^{+}$ and $h_r^{-}(y)=\delta_{1}h_{n} h_{s-r}c^{+}$. By Proposition \ref{prop1}, $h_{r}^{-}(x)>h_{r}^{-}(y)$.
Otherwise, $h_{m}h_{s-r}c^{+}<h_{n}h_{s-r}c^{+}$. We prove the case $m>n$, and the case $m<n$ can be proved similarly. 
Since $m>n$ and $h_{m}h_{s-r}c^{+}<h_{n}h_{s-r}c^{+}$, we can see that the values $h_{z}h_{s-r}c^{+}$, where $z=p, p-1, ..., -q$, are increasing while the index $z$ is decreasing. Thus, for all $z$, $h_{p}h_{s-r}c^{+}\leq h_{z}h_{s-r}c^{+}\leq h_{-q}h_{s-r}c^{+}$.
If $-m<-q$, then by (b), $h_r^{-}(x)=h_{-q}h_{s-r}c^{+}$.
In any case, $h_r^{-}(y)=h_{z}h_{s-r}c^{+}$, for some $z$. Therefore, $h_r^{-}(x)\geq h_r^{-}(y)$. 
Similarly, if $-n>p$, then by (b), $h_r^{-}(y)=h_{p}h_{s-r}c^{+}$; thus, $h_r^{-}(x)\geq h_r^{-}(y)$.
Otherwise, $-q\leq -m<-n\leq p$. By (b), $h_r^{-}(x)= h_{-m}h_{s-r}c^{+}$ and $h_r^{-}(y)=h_{-n}h_{s-r}c^{+}$. Since $-m<-n$, we have $h_{-m}h_{s-r}c^{+}>h_{-n}h_{s-r}c^{+}$, i.e., $h_r^{-}(x)>h_r^{-}(y)$.

(3.2.2) $m=n$. Since $x>y$, by Proposition \ref{prop1}, there exist $k_1,k_2\in H\cup \{I\}$ and $k_1\neq k_2$, and $\sigma_2, \delta_2, \gamma \in H^{*}$ such that $x=\sigma_{2}k_1\gamma h_{m}h_{s}c^{+}, y=\delta_{2}k_2\gamma h_{m}h_{s}c^{+}$, and $k_1\gamma h_{m}h_{s}c^{+}>k_2\gamma h_{m}h_{s}c^{+}$. Also, since $x>y$, we have $m=n\neq 0$ (as a convention, all hedges appearing before $h_0=I$ in a representation of a value have no effect). There are two cases: either $h_{m}h_{s}c^{+}>h_{s}c^{+}$ or $h_{m}h_{s}c^{+}<h_{s}c^{+}$. We prove the case $h_{m}h_{s}c^{+}>h_{s}c^{+}$, and the other can be proved similarly. Since $h_{m}h_{s}c^{+}>h_{s}c^{+}$, by Proposition \ref{prop99}, $Sign(h_{m}h_{s}c^{+})=+1$. There are two cases:

(3.2.2.1) $h_{m}h_{s-r}c^{+}<h_{s-r}c^{+}$. By (b), in any case, $h^{-}_{r}(x)=h^{-}_{r}(y)$.

(3.2.2.2) $h_{m}h_{s-r}c^{+}>h_{s-r}c^{+}$. By (a), $h_{r}^{-}(x)=\sigma_{2}k_1\gamma h_{m}h_{s-r}c^{+}$ and $h_{r}^{-}(y)=\delta_{2}k_2\gamma h_{m}h_{s-r}c^{+}$. Since $h_{m}h_{s-r}c^{+}>h_{s-r}c^{+}$, $Sign(h_{m}h_{s-r}c^{+})=+1=Sign(h_{m}h_{s}c^{+})$. By Definition \ref{def3}, $Sign(k_1\gamma h_{m}h_{s-r}c^{+})=Sign(k_1\gamma h_{m}h_{s}c^{+})$ and $Sign(k_2\gamma h_{m}h_{s-r}c^{+})=Sign(k_2\gamma h_{m}h_{s}c^{+})$. Since $k_1\gamma h_{m}h_{s}c^{+}>k_2\gamma h_{m}h_{s}c^{+}$, there are three cases:

(3.2.2.2.1) $k_1\gamma h_{m}h_{s}c^{+}>k_2\gamma h_{m}h_{s}c^{+}\geq \gamma h_{m}h_{s}c^{+}$. 
Thus, by definition, $k_1>k_2$. Moreover, by Proposition \ref{prop99}, $Sign(k_1\gamma h_{m}h_{s}c^{+})=+1$ and $Sign(k_2\gamma h_{m}h_{s}c^{+})\in \{0,+1\}$. Thus, $Sign(k_1\gamma h_{m}h_{s-r}c^{+})=+1$, i.e., $k_1\gamma h_{m}h_{s-r}c^{+}>\gamma h_{m}h_{s-r}c^{+}$. Since $k_1>k_2$, $k_1\gamma h_{m}h_{s-r}c^{+}\geq k_2\gamma h_{m}h_{s-r}c^{+}\geq \gamma h_{m}h_{s-r}c^{+}$; thus, $h^{-}_{r}(x)\geq h^{-}_{r}(y)$.

(3.2.2.2.2) $\gamma h_{m}h_{s}c^{+}\geq k_1\gamma h_{m}h_{s}c^{+}>k_2\gamma h_{m}h_{s}c^{+}$. The proof is similar to that of (3.2.2.2.1).

(3.2.2.2.3) $k_1\gamma h_{m}h_{s}c^{+}\geq \gamma h_{m}h_{s}c^{+}\geq k_2\gamma h_{m}h_{s}c^{+}$. By Proposition \ref{prop99}, $Sign(k_1\gamma h_{m}h_{s}c^{+})$ $=Sign(k_1\gamma h_{m}h_{s-r}c^{+})\in \{0,+1\}$ and $Sign(k_2\gamma h_{m}h_{s}c^{+}) = Sign(k_2\gamma h_{m}h_{s-r}c^{+})\in \{0,-1\}$. Thus, $k_1\gamma h_{m}h_{s-r}c^{+}\geq \gamma h_{m}h_{s-r}c^{+}$ and $k_2\gamma h_{m}h_{s-r}c^{+}\leq \gamma h_{m}h_{s-r}c^{+}$. 
Since $k_1\gamma h_{m}h_{s}c^{+}>k_2\gamma h_{m}h_{s}c^{+}$, one of $Sign(k_1\gamma h_{m}h_{s-r}c^{+})$ and 
$Sign(k_2\gamma h_{m}h_{s-r}c^{+})$ must differ from 0; thus $k_1\gamma h_{m}h_{s-r}c^{+}> k_2\gamma h_{m}h_{s-r}c^{+}$. Therefore, $h^{-}_{r}(x)> h^{-}_{r}(y)$.

Second, consider the case $x,y\in H(c^{-})$. In any case, $h_{r}^{-}(x)=-h_{z}^{-}(-x)$ and 
$h_{r}^{-}(y)=-h_{z}^{-}(-y)$, for some $z$. Since $x,y\in H(c^{-})$, we have $-x,-y\in H(c^{+})$. By the above case, 
$x>y\Rightarrow -x<-y \Rightarrow h_{z}^{-}(-x)\leq h_{z}^{-}(-y)\Rightarrow h_{r}^{-}(x)\geq h_{r}^{-}(y)$.

Finally, if $x\in H(c^{+})\cup \{W,1\}$ and $y\in H(c^{-})\cup \{0,W\}$, then $h^{-}(x)\geq W\geq h^{-}(y)$; and if $x=1$, then $h^{-}(x)\geq h^{-}(y)$.
\end{proof}

\subsection{Limited hedge algebras}
In the present work, we only deal with finite linguistic truth domains. The rationale for this is as follows.

First, in daily life, humans only use linguistic terms with a limited length. This is due to the fact that it is difficult to distinguish the different meaning of terms with many hedges such as \emph{Very Little Probably True} and \emph{More Little Probably True}. Hence, we can assume that applying any hedge to truth values that have a certain number $l$ of hedges will not change their meaning. In other words, canonical representations of all terms w.r.t. primary terms have a length of at most $l+1$.

Second, according to \citeN{Zadeh75b}, in most applications to approximate reasoning, a small finite set of fuzzy truth values would, in general, be sufficient since each fuzzy truth value represents a fuzzy set rather than a single element of [0,1].

Third, more importantly, it is reasonable for us to consider only finitely many truth values in order to provide a logical system that can be implemented for computers. In fact, we later show that with a finite truth domain, we can obtain the Least Herbrand model for a finite program after a finite number of iterations of an immediate consequences operator. 

\begin{definition}[l-limited HA]
An \emph{l-limited} HA, where $l$ is a positive integer, is a lin-HA in which canonical representations of all terms w.r.t. primary terms have a length of at most $l+1$.
\end{definition}
For an \emph{l}-limited HA $\underline{X} = (X,G,H,\leq)$, since the set of hedges $H$ is finite, so is the linguistic truth domain $\overline{X}$.

In the following, we give a particular example of inverse mappings of hedges for a 2-limited HA.

\begin{example} \label{ex99}
Consider a $2$-limited HA $\underline{X}=(X,\{c^{+}, c^{-}\},\{V,M,P,L\},\leq)$ with $L<_e P<_e I<_e M<_e V$. We have a linguistic truth domain $\overline{X}=\{
v_0=0, 
v_1=VVc^{-},v_2=MVc^{-}, v_3=Vc^{-}, v_4=PVc^{-}, v_5=LVc^{-},
v_6=VMc^{-}, v_7=MMc^{-}, v_8=Mc^{-}, v_9=PMc^{-}, v_{10}=LMc^{-}, 
v_{11}=c^{-},
v_{12}=VPc^{-}, v_{13}=MPc^{-}, v_{14}=Pc^{-}, v_{15}=PPc^{-}, 
v_{16}=LPc^{-}, 
v_{17}=LLc^{-}, v_{18}=PLc^{-}, v_{19}=Lc^{-}, v_{20}=MLc^{-}, 
v_{21}=VLc^{-}, 
v_{22}=W,
v_{23}=VLc^{+}, v_{24}=MLc^{+}, v_{25}=Lc^{+}, v_{26}=PLc^{+}, 
v_{27}=LLc^{+}, 
v_{28}=LPc^{+}, v_{29}=PPc^{+}, v_{30}=Pc^{+}, v_{31}=MPc^{+}, 
v_{32}=VPc^{+}, 
v_{33}=c^{+},
v_{34}=LMc^{+}, v_{35}=PMc^{+}, v_{36}=Mc^{+}, v_{37}=MMc^{+}, 
v_{38}=VMc^{+}, 
v_{39}=LVc^{+}, v_{40}=PVc^{+}, v_{41}=Vc^{+}, v_{42}=MVc^{+}, 
v_{43}=VVc^{+},
v_{44}=1\}$.

Based on the inverse mappings defined in Example \ref{ex3}, we can build the inverse mappings for this 2-limited HA with some modifications.
Since we are working with the 2-limited HA, if $h^{-}(x)=W$, for $x\in H(c^{+})$, we can put $h^{-}(x)=VLc^{+}$, the minimum value of $H(c^{+})$; if $h^{-}(x)=1$, for $x\in H(c^{+})$, we can put $h^{-}(x)=VVc^{+}$, the maximum value of $H(c^{+})$; if $h^{-}(x)=W$, for $x\in H(c^{-})$, we can put $h^{-}(x)=VLc^{-}$, the maximum value of $H(c^{-})$; and if $h^{-}(x)=0$, for $x\in H(c^{-})$, we can put $h^{-}(x)=VVc^{-}$, the minimum value of $H(c^{-})$. 
Changes are also made to the inverse mappings of hedges with a value in $\{c^{-},c^{+}\}$. 
This means that inverse mappings of hedges are not unique. This is acceptable since reasoning based on fuzzy logic is approximate, and inverse mappings of hedges should be built according to applications. 

Inverse mappings of hedges for the 2-limited HA are shown in Table \ref{tab3}, in which the value of an inverse mapping of a hedge $h^{-}$, appearing in the first row, of a value $x$, appearing in the first column, is in the corresponding cell. For example, $M^{-}(PPc^{+})=MLc^{+}$. Note that the values of $x$ appear in an ascending order.

\begin{table}
\caption{Inverse mappings of hedges}
\label{tab3}
\begin{minipage}{\textwidth}
\begin{tabular}{lccccc}
\hline\hline
& $V^{-}$ & $M^{-}$ & $P^{-}$ & $L^{-}$ & \\
\hline
$0$ & $0$ & $0$ & $0$ & $0$&\\
$kVc^{-}$ & $VVc^{-}$ & $VVc^{-}$ & $kMc^{-}$ & $c^{-}$&$^{a}$\\ 
$kMc^{-}$ & $VVc^{-}$ & $kVc^{-}$ & $c^{-}$ & $kPc^{-}$&$^{a}$\\
$c^{-}$ & $Vc^{-}$ & $Mc^{-}$ & $Pc^{-}$ & $Lc^{-}$&\\
$VPc^{-}$ & $VMc^{-}$ & $PMc^{-}$ & $LLc^{-}$ & $VLc^{-}$&\\
$MPc^{-}$ & $MMc^{-}$ & $LMc^{-}$ & $PLc^{-}$ & $VLc^{-}$&\\
$Pc^{-}$ & $Mc^{-}$ & $c^{-}$ & $Lc^{-}$ & $VLc^{-}$&\\ 
$PPc^{-}$ & $PMc^{-}$ & $VPc^{-}$ & $MLc^{-}$ & $VLc^{-}$&\\
$LPc^{-}$ & $LMc^{-}$ & $VPc^{-}$ & $VLc^{-}$ & $VLc^{-}$&\\
$LLc^{-}$ & $LMc^{-}$ & $VPc^{-}$ & $VLc^{-}$ & $VLc^{-}$&\\
$PLc^{-}$ & $LMc^{-}$ & $MPc^{-}$ & $VLc^{-}$ & $VLc^{-}$&\\
$Lc^{-}$ & $c^{-}$ & $Pc^{-}$ & $VLc^{-}$ & $VLc^{-}$&\\ 
$MLc^{-}$ & $VPc^{-}$ & $PPc^{-}$ & $VLc^{-}$ & $VLc^{-}$&\\ 
$VLc^{-}$ & $PPc^{-}$ & $LPc^{-}$ & $VLc^{-}$ & $VLc^{-}$&\\ 
$W$ & $W$ & $W$ & $W$ & $W$&\\
$VLc^{+}$ & $VLc^{+}$ & $VLc^{+}$ & $LPc^{+}$ & $PPc^{+}$&\\ 
$MLc^{+}$ & $VLc^{+}$ & $VLc^{+}$ & $PPc^{+}$ & $VPc^{+}$&\\ 
$Lc^{+}$ & $VLc^{+}$ & $VLc^{+}$ & $Pc^{+}$ & $c^{+}$&\\ 
$PLc^{+}$ & $VLc^{+}$ & $VLc^{+}$ & $MPc^{+}$ & $LMc^{+}$&\\ 
$LLc^{+}$ & $VLc^{+}$ & $VLc^{+}$ & $VPc^{+}$ & $LMc^{+}$&\\ 
$LPc^{+}$ & $VLc^{+}$ & $VLc^{+}$ & $VPc^{+}$ & $LMc^{+}$&\\
$PPc^{+}$ & $VLc^{+}$ & $MLc^{+}$ & $VPc^{+}$ & $PMc^{+}$&\\
$Pc^{+}$ & $VLc^{+}$ & $Lc^{+}$ & $c^{+}$ & $Mc^{+}$&\\ 
$MPc^{+}$ & $VLc^{+}$ & $PLc^{+}$ & $LMc^{+}$ & $MMc^{+}$&\\ 
$VPc^{+}$ & $VLc^{+}$ & $LLc^{+}$ & $PMc^{+}$ & $VMc^{+}$&\\ 
$c^{+}$ & $Lc^{+}$ & $Pc^{+}$ & $Mc^{+}$ & $Vc^{+}$&\\
$kMc^{+}$ & $kPc^{+}$ & $c^{+}$ & $kVc^{+}$ & $VVc^{+}$&$^{a}$\\
$kVc^{+}$ & $c^{+}$ & $kMc^{+}$ & $VVc^{+}$ & $VVc^{+}$& \footnote{ $k$ is any of the hedges, including the identity $I$.}\\ 
$1$ & $1$ & $1$ & $1$ & $1$&\\
\hline\hline
\end{tabular}
\vspace{-2\baselineskip}
\end{minipage}
\end{table}
\end{example}

\subsection{Many-valued modus ponens}
Our logic is truth-functional, i.e., the truth value of a compound formula, built from its components using a logical connective, is a function, which is called the \emph{truth function} of the connective, of the truth values of the components. 

Our procedural semantics is developed based on many-valued modus ponens.
In order to guarantee the soundness of many-valued modus ponens, the truth function of an implication, called an \emph{implicator}, must be \emph{residual} to the \emph{t-norm}, a commutative and associative binary operation on the truth domain, evaluating many-valued modus ponens \cite{Ha98}. 
The many-valued modus ponens syntactically looks like: 
\[\frac{(B,b), (A\leftarrow B,r)}{(A,\mathcal{C}(b,r))}\]
Its soundness semantically states that whenever $f$ is an interpretation such that $f(B)\geq b$, i.e., $f$ is a model of $(B,b)$, and $f(A\leftarrow B)=\leftarrow^{\bullet}(f(A),f(B))\geq r$, i.e., $f$ is a model of $(A\leftarrow B,r)$, then $f(A)\geq \mathcal{C}(b,r)$, where $\leftarrow^{\bullet}$ is an implicator, and $\mathcal{C}$ is a t-norm.
This means the truth value of $A$ under any model of $(B,b)$ and $(A\leftarrow B,r)$ is at least $\mathcal{C}(b,r)$.
More precisely, let $r$ be a lower bound to the truth value of the implication $h\leftarrow b$, let $\mathcal{C}$ be a t-norm, and let $\leftarrow^\bullet$ be its residual implicator; we have: 
\begin{equation}
\mathcal{C}(b,r)\leq h \mbox{ iff } r\leq \leftarrow^\bullet(h,b) \label{ti01}
\end{equation}
According to \citeN{Ha98}, from (\ref{ti01}), we have:
\begin{eqnarray}
(\forall b)(\forall h)\; \mathcal{C}(b,\leftarrow^\bullet(h,b))\leq h \label{ti02} \\
(\forall b)(\forall r)\; \leftarrow^\bullet(\mathcal{C}(b,r),b)\geq r \label{ti03}
\end{eqnarray}
Note that t-norms are not necessary to be a truth function of any conjunction in our language.

Recall that in many-valued logics, there are several prominent sets of connectives called \L ukasiewicz, G\"{o}del, and product logic ones. Each of the sets has a pair of residual t-norm and implicator. Since our truth values are linguistic, we cannot use the product logic connectives. 

Given a linguistic truth domain $\overline{X}$, since all the values in $\overline{X}$ are linearly ordered, we assume that they are $v_0\leq v_1\leq ...\leq v_n$, where $v_0=0$ and $v_n=1$. The \L ukasiewicz t-norm and implicator can be defined on $\overline{X}$ as follows:
\[\mathcal{C}_{L}(v_i,v_j)= \left\{\begin{array}{ll}
                                                 v_{i+j-n}	& \mbox{if $i+j-n> 0$} \\
                                                 v_0	& \mbox{otherwise} 
                                                 \end{array}
                                            \right.\]
\[\leftarrow_{L}^\bullet (v_j,v_i)= \left\{\begin{array}{ll}
                                                 v_n	& \mbox{if $i\leq j$} \\
                                                 v_{n+j-i}	& \mbox{otherwise}
                                                 \end{array}
                                            \right.\]
and those of G\"{o}del can be:
\[\mathcal{C}_{G} (v_i,v_j)= min(v_i,v_j)\]
\[\leftarrow_{G}^\bullet (v_j,v_i)= \left\{\begin{array}{ll}
                                                 v_n	& \mbox{if $i\leq j$} \\
                                                 v_j	& \mbox{otherwise}
                                                 \end{array}
                                            \right.\]
Clearly, each of the implicators is the residuum of the corresponding t-norm. It can also be seen that t-norms are monotone in all arguments, and implicators are non-decreasing in the first argument and non-increasing in the second.

\section{Fuzzy linguistic logic programming}
\subsection{Language}
Like \citeN{Vo01}, our language is a many sorted (typed) predicate language. Let $\mathcal{A}$ denote the set of all attributes. For each sort of variables $A\in\mathcal{A}$, there is a set $\mathcal{C}^A$ of constant symbols, which are names of elements of the domain of $A$. In order to achieve the Least Herbrand model after a finite number of iterations of an immediate consequences operator, we do not allow any function symbols. This is not a severe restriction since in many database applications, there are no function symbols involved. 

Connectives can be: conjunctions $\wedge$ (also called G\"{o}del) and $\wedge_{L}$ (\L ukasiewicz); the disjunction $\vee$; implications $\leftarrow_L$ (\L ukasiewicz) and  $\leftarrow_G$ (G\"{o}del); and linguistic hedges as unary connectives. For any connective $c$ different from hedges, its truth function is denoted by $c^\bullet$, and for a hedge connective $h$, its truth function is its inverse mapping $h^{-}$. The only quantifier allowed is the universal quantifier $\forall$. 

A \emph{term} is either a constant or a variable. 

An \emph{atom} or \emph{atomic formula} is of the form $p(t_1,...,t_n)$, where $p$ is an n-ary predicate symbol, and $t_1,...,t_n$ are terms of corresponding attributes $A_1,...,A_n$. 

A \emph{body formula} is defined inductively as follows:
($i$) An atom is a body formula.
($ii$) If $B_1$ and $B_2$ are body formulae, then so are $\wedge(B_1,B_2)$, $\vee(B_1,B_2)$, and $hB_1$, where $h$ is a hedge. Here, we use the prefix notation for connectives in body formulae. 

A \emph{rule} is a graded implication $(A\leftarrow B.r)$, where $A$ is an atom called \emph{rule head}, $B$ is a body formula called \emph{rule body}, and $r$ is a truth value different from 0. $(A\leftarrow B)$ is called the \emph{logical part} of the rule. 

A \emph{fact} is a graded atom ($A.b$), where $A$ is an atom called the logical part of the fact, and $b$ is a truth value different from 0.

\begin{definition}[Fuzzy linguistic logic program]
A \emph{fuzzy linguistic logic program} (program, for short) is a finite set of rules and facts, where truth values are from the linguistic truth domain of an $l$-limited HA, hedges used in body formulae (if any) belong to the set of hedges of the HA, and there are no two rules (facts) having the same logical part, but different truth values.
\end{definition}
We follow Prolog conventions where predicate symbols and constants begin with a lower-case letter, and variables begin with a capital letter. 
\begin{example} \label{ex101}
Assume we use the truth domain from the 2-limited HA in Example \ref{ex99}, that is, $\underline{X} = (X,\{False,True\},\{V, M, P,L\},\leq)$, and we have the following knowledge base:

($i$) The sentence ``\emph{If a student studies very hard, and his/her university is probably high-ranking, then he/she will be a good employee}" is \emph{Very More True}.

($ii$) The sentence ``\emph{The university where Ann is studying is high-ranking}" is \emph{Very True}.

($iii$) The sentence ``\emph{Ann is studying hard}" is \emph{More True}.

Let \emph{gd\_em, st\_hd, hira\_un}, and \emph{T} stand for ``\emph{good employee}", ``\emph{study hard}", ``\emph{high-ranking university}", and ``\emph{True}", respectively. Then, the knowledge base can be represented by the following program: 
\begin{eqnarray*}
(gd\_em(X)\leftarrow_G \wedge(V\;st\_hd(X),P\;hira\_un(X)).VMT)\\
(hira\_un(ann).VT)\\
(st\_hd(ann).MT)
\end{eqnarray*}
Note that the predicates $st\_hd(X)$ and $hira\_un(X)$ in the only rule are modified by the hedges $V$ and $P$, respectively.
\end{example}
We assume as usual that the underlying language of a program $P$ is defined by constants (if no such constant exists, we add some constant such as $a$ to form ground terms) and predicate symbols appearing in $P$. With this understanding, we can now refer to the \textit{Herbrand universe} of sort $A$, which consists of all ground terms of $A$, by $U^A_P$, and to the \textit{Herbrand base} of $P$, which consists of all ground atoms, by $B_P$ \cite{Ll87}.

A program $P$ can be represented as a partial mapping:
\[P:Formulae\rightarrow \overline{X} \setminus \{0\}\]
\noindent where the domain of $P$, denoted by $dom(P)$, is finite and consists only of logical parts of rules and facts, and $\overline{X}$ is a linguistic truth domain. The truth value of a rule ($A\leftarrow B.r$) is $r=P(A\leftarrow B)$, and that of a fact ($A.b$) is $b=P(A)$.

Since in our logical system we only want to obtain the computed answers for queries, we do not look for 1-tautologies to extend the capabilities of the system although we can have some due to the fact that our connectives are classical many-valued ones (see \citeN{Ha98}). 

\subsection{Declarative semantics}
Since we are working with logic programs without negation, it is reasonable to consider only fuzzy Herbrand interpretations and models. Given a program $P$, let $\overline{X}$ be the linguistic truth domain; a \emph{fuzzy linguistic Herbrand interpretation} (interpretation, for short) $f$ is a mapping $f:B_P\rightarrow \overline{X}$. The ordering $\leq$ in $\overline{X}$ can be extended to the set of interpretations as follows. We say $f_1\sqsubseteq f_2 $ iff $f_1(A)\leq f_2(A)$ for all ground atoms $A$. Clearly, the set of all interpretations of a program is a complete lattice under $\sqsubseteq$. The least interpretation called the \emph{bottom interpretation}, denoted by $\bot$, maps every ground atom to 0.

An interpretation $f$ can be extended to all ground formulae, denoted by $\overline{f}$, using the unique homomorphic extension as follows: 
($i$) $\overline{f}(A)=f(A)$, if $A$ is a ground atom;
($ii$) $\overline{f}(c(B_1,B_2))=c^\bullet(\overline{f}(B_1),\overline{f}(B_2))$, where $B_1,B_2$ are ground formulae, and $c$ is a binary connective;
($iii$) $\overline{f}(hB)=h^{-}(\overline{f}(B))$, where $B$ is a ground body formula, and $h$ is a hedge. 

For non-ground formulae, since all the formulae in the language are considered universally quantified, the interpretation $\overline{f}$ is defined as
\[\overline{f}(\varphi)=\overline{f}(\forall\varphi)=inf_{\vartheta}\{\overline{f}(\varphi\vartheta)|\varphi\vartheta \mbox{ is a ground instance of } \varphi\}\]
where $\forall\varphi$ means universal quantification of all variables with free occurrence in $\varphi$. 

An interpretation $f$ is a \emph{model} of a program $P$ if for all formulae $\varphi\in dom(P)$, we have $\overline{f}(\varphi)\geq P(\varphi)$. Therefore, $P(\varphi)$ is understood as a lower bound to the truth value of $\varphi$.

A \emph{query} is an atom used as a question $?A$ prompting the system.

\begin{definition}[Correct answer]
Given a program $P$, let $\overline{X}$ be the linguistic truth domain. A pair $(x;\theta)$, where $x\in \overline{X}$, and $\theta$ is a substitution, is called a \emph{correct answer} for $P$ and a query $?A$ if for any model $f$ of $P$, we have $\overline{f}(A\theta)\geq x$.
\end{definition}

\subsection{Procedural semantics}
Given a program $P$ and a query $?A$, we want to compute a lower bound for the truth value of $A$ under any model of $P$.  Recall that in the theory of many-valued modus ponens \cite{Ha98}, given $(A\leftarrow B.r)$ and $(B.b)$, we have $(A.\mathcal{C}(b,r))$. As in \citeN{Vo01}, our procedural semantics utilises \emph{admissible rules}. 

Admissible rules act on tuples of words in the alphabet, denoted by $L^{e}_P$, which is the disjoint union of the alphabet of the language of $dom(P)$ augmented by the truth functions of the connectives (except $\leftarrow_i$ and $\leftarrow^{\bullet}_i$) and symbols $\mathcal{C}_i$, and the linguistic truth domain. 

\begin{definition}[Admissible rules]
Admissible rules are defined as follows: \\
\textbf{Rule 1.} From $((X A_m Y);\vartheta)$ infer $((X \mathcal{C}(B,r) Y)\theta;\vartheta \theta)$ if 
\begin{enumerate} 
\item $A_m$ is an atom (called \textit{the selected atom}) 
\item $\theta$ is an mgu of $A_m$ and $A$ 
\item $(A\leftarrow B.r)$ is a rule in the program.
\end{enumerate}
\textbf{Rule 2.} From $(X A_m Y)$ infer $(X 0 Y)$. This rule is usually used for situations where $A_m$ does not unify with any rule head or logical part of facts in the program. \\
\textbf{Rule 3.} From $(X hB Y)$ infer $(X h^{-}(B) Y)$ if $B$ is a non-empty body formula, and $h$ is a hedge.\\
\textbf{Rule 4.} From $((X A_m Y);\vartheta)$ infer $((X r Y)\theta;\vartheta\theta)$ if 
\begin{enumerate} 
\item $A_m$ is an atom (also called the selected atom)
\item $\theta$ is an mgu of $A_m$ and $A$ 
\item $(A.r)$ is a fact in the program.
\end{enumerate}
\textbf{Rule 5.} If there are no more predicate symbols in the word, replace all connectives $\wedge$'s, and $\vee$'s with $\wedge^{\bullet}$, and $\vee^{\bullet}$, respectively. Then, since this word contains only some additional $\mathcal{C}$'s, $h^{-}$'s, and truth values, evaluate it. The substitution remains unchanged.
\end{definition}
Note that our rules except Rule 3 are the same as those in \citeN{Vo01}.

\begin{definition}[Computed answer]
Let $P$ be a program and $?A$ a query. A pair $(r;\theta)$, where $r$ is a truth value, and $\theta$ is a substitution, is said to be a \emph{computed answer} for $P$ and $?A$ if there is a sequence $G_0,...,G_n$ such that
\begin{enumerate}
\item every $G_i$ is a pair consisting of a word in $L^e_P$ and a substitution
\item $G_0=(A;id)$
\item every $G_{i+1}$ is inferred from $G_i$ by one of the admissible rules (here we also utilise the usual Prolog renaming of variables along derivation)
\item $G_n=(r;\theta')$ and $\theta=\theta'$ restricted to variables of $A$,
\end{enumerate}
and we say that the \textit{computation} has a length of $n$.
\end{definition}
Let us give an example of a computation.

\begin{example} \label{ex4}
We take the program in Example \ref{ex101}, that is:
\begin{eqnarray*}
(gd\_em(X)\leftarrow_G \wedge(Vst\_hd(X),Phira\_un(X)).VMT)\\
(hira\_un(ann).VT)\\
(st\_hd(ann).MT)
\end{eqnarray*}
Given a query $?gd\_em(ann)$, we can have the following computation (since the query is ground, the substitution in the computed answer is the identity):
\begin{eqnarray*}
?gd\_em(ann)\\
\mathcal{C}_G(\wedge(V\;st\_hd(ann), P\;hira\_un(ann)),V M T)\\
\mathcal{C}_G(\wedge(V^{-}(st\_hd(ann)), P\;hira\_un(ann)),V M T)\\
\mathcal{C}_G(\wedge(V^{-}(st\_hd(ann)), P^{-}(hira\_un(ann))),V M T)\\
\mathcal{C}_G(\wedge(V^{-}(M T), P^{-}(hira\_un(ann))),V M T)\\
\mathcal{C}_G(\wedge(V^{-}(M T), P^{-}(V T)),V M T)\\
\mathcal{C}_G(\wedge^\bullet(V^{-}(M T), P^{-}(V T)),V M T)
\end{eqnarray*}
Using the inverse mappings of hedges in Table \ref{tab3}, we have
$\mathcal{C}_G(\wedge^\bullet(V^{-}(M T),P^{-}(V T))$, $V M T)=\mathcal{C}_G(min(P T, VV T),$ $V M T)=\mathcal{C}_G(P T,V M T)=PT$.
Hence, the sentence ``\emph{Ann will be a good employee}" is at least \emph{Probably True}. 
This result is reasonable as follows: one of the conditions constituting the result is the one saying that ``\emph{The student studies very hard}"; since ``\emph{Ann is studying hard}" is \emph{MT} (\emph{More True}), the truth value of ``\emph{Ann is studying very hard}" is $V^{-}(M T)$; and since $MT<VT$, we have $V^{-}(M T)<V^{-}(V T)=T$, and $V^{-}(M T)=PT$ is acceptable. 

If we use the \L ukasiewicz implication instead of the G\"{o}del implication in the rule, then in the computation, the G\"{o}del t-norm will be replaced by the \L ukasiewicz t-norm, and, finally, we have an answer $(gd\_em(ann).MLT)$. 
\end{example}
From the definition of the procedural semantics, we can see that in order to increase the chances of finding a good computed answer which has a better truth value along a computation, we should do the following:

$(i)$ If there is more than one rule or fact whose rule heads or logical parts can be unifiable with the selected atom, and of such rules or facts there is only one to which the highest truth value is assigned, then we choose it for the next step. 

$(ii)$ If there is one fact among such rules or facts which are associated with the highest truth value, then we choose the fact for the next step since the t-norm evaluating such a rule always yields a lower truth value than that of the fact. 

$(iii)$ If there is more than one such a rule, but no facts, which have the highest truth value, then we choose the one with the G\"{o}del implication for the next step since in this case, the G\"{o}del t-norm usually, but not always (since it  also depends on the bodies of the rules), yields a better truth value than the \L ukasiewicz t-norm. In Example \ref{ex4}, it has been shown that with the same body formula, the rule with the G\"{o}del implication yields a better result ($PT$) than the rule with the \L ukasiewicz implication ($MLT$).

\subsection{Soundness of the procedural semantics}
\begin{theorem} \label{th04}
Every computed answer for a program $P$ and a query $?A$ is a correct answer for $P$ and $?A$.
\end{theorem}
\noindent
\begin{proof}
Assume that a pair $(r;\theta)$ is a computed answer for $P$ and $?A$. Let $f$ be any model of $P$; we will prove that $\overline{f}(A\theta)\geq r$.

The proof is by induction on length $n$ of computations.

First, suppose that $n=1$. Hence, either Rule 2 or Rule 4 has been applied. The case of Rule 2 is obvious since $r=0$. The case of Rule 4 implies that $P$ has a fact $(C.r)$ such that $A\theta=C\theta$. Therefore, $\overline{f}(A\theta)=\overline{f}(C\theta)\geq \overline{f}(C)\geq P(C)=r$.

Next, suppose that the result holds for computed answers coming from computations of length $\leq k-1$, where $k>1$. We prove that it also holds for a computation of length $k$. 

Assume that the sequence of the substitutions in the computation is $\theta_1,...,\theta_k$ (some of them are the identity), where $\theta=\theta_1...\theta_k$ restricted to variables of $A$. Since the length of the computation $k>1$, the first admissible rule to be applied is Rule 1. This means there exists a rule $(C\leftarrow_i B.c)$ in $P$ such that $A\theta_1=C\theta_1$. For each atom $D$ in the rule body $B\theta_1$, there exists a computation of length $\leq k-1$ for it. Suppose $d$ is the computed truth value for $D$ in that computation; by the induction hypothesis, we have $d\leq \overline{f}(D\theta_2...\theta_k)$. Furthermore, since the truth functions of the conjunctions, the disjunction, and inverse mappings of hedges are non-decreasing in all their arguments, if $b$ is the computed truth value for the whole rule body $B\theta_1$, which is calculated from all the $d$ for each atom $D$ using the truth functions of the connectives, then $b\leq \overline{f}(B\theta_1\theta_2...\theta_k)$. Therefore, we have: $r=\mathcal{C}_i(b,c)\leq \mathcal{C}_i(\overline{f}(B\theta_1...\theta_k),c)\leq^{(*)} \mathcal{C}_i(\overline{f}(B\theta_1...\theta_k),\overline{f}(C\theta_1...\theta_k\leftarrow_i B\theta_1...\theta_k))= \mathcal{C}_i(\overline{f}(B\theta_1...\theta_k),\leftarrow_i^\bullet(\overline{f}(C\theta_1...\theta_k), \overline{f}(B\theta_1...\theta_k))) \leq^{(**)} \overline{f}(C\theta_1...\theta_k)=\overline{f}(A\theta_1...\theta_k)=\overline{f}(A\theta)$,
where (*) holds since $f$ is a model of $P$, and (**) follows from (\ref{ti02}).
\end{proof}

\subsection{Fixpoint semantics}

Similar to \citeN{kra04}, the immediate consequences operator, introduced by van Emden and Kowalski, can be generalised to the case of fuzzy linguistic logic programming as follows.

\begin{definition}[Immediate consequences operator] \label{ico}
Let $P$ be a program. The operator $T_P$ mapping from interpretations to interpretations is defined as follows. For every interpretation $f$ and every ground atom $A\in B_P$,

$T_P(f)(A)=max\{sup\{\mathcal{C}_i(\overline{f}(B),r):(A\leftarrow_i B.r)$ is a ground instance of a rule in $P\}$, $sup\{b:(A.b)$ is a ground instance of a fact in $P\}\}$.
\end{definition}
Since $P$ is function-free, each Herbrand universe $U_P^A$ of a sort $A$ is finite, and so is its Herbrand base $B_P$. Hence, for each $A\in B_P$, there are a finite number of ground instances of rule heads and logical parts of facts which match $A$. Therefore, the suprema in the definition of $T_P$ are in fact maxima.

Similar to \citeN{Medina04}, we have the following results. 
\begin{theorem} \label{th01}
The operator $T_P$ is monotone.
\end{theorem}
\noindent
\begin{proof}
Let $f_1$ and $f_2$ be two interpretations such that $f_1\sqsubseteq f_2$; we prove that $T_P(f_1)\sqsubseteq T_P(f_2)$.

First, let us prove $\overline{f_1}(B)\leq \overline{f_2}(B)$ for all ground body formulae $B$ by induction on the structure of the formulae. In the base case where $B$ is a ground atom, we have $\overline{f_1}(B)=f_1(B)\leq f_2(B)=\overline{f_2}(B)$. For the inductive case, consider a ground body formula $B$. By case analysis and the induction hypothesis, we have $B=\wedge(B_1, B_2)$, or $B=\vee(B_1, B_2)$, or $B=hB_1$ such that $\overline{f_1}(B_1)\leq \overline{f_2}(B_1)$ and $\overline{f_1}(B_2)\leq \overline{f_2}(B_2)$. By definition, we have $\overline{f_1}(B)=\wedge^{\bullet}(\overline{f_1}(B_1),\overline{f_1}(B_2))\leq \wedge^{\bullet}(\overline{f_2}(B_1),\overline{f_2}(B_2))=\overline{f_2}(B)$, or $\overline{f_1}(B)=\vee^{\bullet}(\overline{f_1}(B_1),\overline{f_1}(B_2))\leq \vee^{\bullet}(\overline{f_2}(B_1),\overline{f_2}(B_2))=\overline{f_2}(B)$, or
$\overline{f_1}(B)$ = $h^{-}(\overline{f_1}(B_1))\leq h^{-}(\overline{f_2}(B_1))=\overline{f_2}(B)$, respectively. Thus, $\overline{f_1}(B)\leq \overline{f_2}(B)$ for all ground body formulae $B$.

Now, let $A$ be any ground atom. If $A$ does not unify with any rule head or logical part of facts in $P$, then $T_P(f_1)(A)=T_P(f_2)(A)=0$. Otherwise, since the value of the second $sup$ in Definition \ref{ico} does not depend on the interpretations, what we need to consider now is the first $sup$. For any ground instance $(A\leftarrow_i B.r)$ of a rule in $P$, since $B$ is ground, we have $\mathcal{C}_i(\overline{f_1}(B),r)\leq \mathcal{C}_i(\overline{f_2}(B),r)$. By taking suprema for all ground instances $(A\leftarrow_i B.r)$ on both sides, we have $sup\{\mathcal{C}_i(\overline{f_1}(B),r)\}\leq sup\{\mathcal{C}_i(\overline{f_2}(B),r)\}$. Therefore, $T_P(f_1)(A)\leq T_P(f_2)(A)$ for all ground atoms $A$.
\end{proof}
\begin{theorem} \label{th02}
The operator $T_P$ is continuous.
\end{theorem}
\noindent
\begin{proof}
Recall that a mapping $f:L\rightarrow L$, where $L$ is a complete lattice, is said to be \emph{continuous} if for every directed subset $X$ of $L$, $f(sup(X))=sup\{f(x)|x\in X\}$.

Let us prove that for each directed set $X$ of interpretations, $T_P(sup (X))=sup\{T_P(f)|f\in X\}$.

Since $T_P$ is monotone, we have $sup\{T_P(f)|f\in X\}\sqsubseteq T_P(sup (X))$. On the other hand, since the Herbrand base $B_P$ and the truth domain are finite, the set of all Herbrand interpretations of $P$ is finite. Therefore, for each finite directed set $X$ of interpretations, we have an upper bound of $X$ in $X$. This, together with the monotonicity of $T_P$, leads to $T_P(sup(X))\sqsubseteq sup\{T_P(f):f\in X\}$.
\end{proof}
\begin{theorem} \label{th03}
An interpretation $f$ is a model of a program $P$ iff $T_P(f)\sqsubseteq f$.
\end{theorem}
\noindent
\begin{proof}
First, assume that $f$ is a model of $P$; we prove that $T_P(f)\sqsubseteq f$.

Let $A$ be any ground atom. Consider the following cases:

($i$) If $A$ is neither a ground instance of a logical part of facts nor a ground instance of a rule head in $P$, then $T_P(f)(A)=0\leq f(A)$.

($ii$) For each ground instance $(A.b)$ of a fact, say $(C.b)$, in $P$, since $f$ is a model of $P$, and $A$ is a ground instance of $C$, we have $b=P(C)\leq \overline{f}(C)\leq f(A)$. Hence, $f(A)\geq sup\{b|(A.b)$ is a ground instance of a fact in $P\}$. 

($iii$) For each ground instance $(A\leftarrow_i B.r)$ of a rule, say $(C.r)$, in $P$, we have:
$\mathcal{C}_i(\overline{f}(B),r)=\mathcal{C}_i(\overline{f}(B),P(C))\leq^{(*)} \mathcal{C}_i(\overline{f}(B),\overline{f}(A\leftarrow_i B))=\mathcal{C}_i(\overline{f}(B),\leftarrow_i^\bullet(f(A),\overline{f}(B)))\leq^{(**)}f(A)$,
where (*) holds since $(A\leftarrow_i B)$ is a ground instance of $C$, and (**) follows from (\ref{ti02}). Therefore, $f(A)\geq sup\{\mathcal{C}_i(\overline{f}(B),r)|(A\leftarrow_i B.r)$ is a ground instance of a rule in $P\}$. 

Thus, by definition, $T_P(f)(A)\leq f(A)$ for all $A\in B_P$.

Finally, let us show that if $T_P(f)\sqsubseteq f$, then $f$ is a model of $P$.

Let $C$ be any formula in $dom(P)$. There are two cases:

($i$) $(C.c)$, where $c$ is a truth value, is a fact in $P$. For each ground instance $A$ of $C$, by hypothesis and definition, we have $f(A)\geq T_P(f)(A)\geq sup\{b|(A.b)$ is a ground instance of a fact in $P\}\geq c=P(C)$. Therefore, $\overline{f}(C)=inf\{f(A)|A$ is a ground instance of $C\}\geq P(C)$.

($ii$) $(C.c)$ is a rule in $P$. For each ground instance $A\leftarrow_j D$ of $C$, by hypothesis and definition, we have $f(A)\geq T_P(f)(A)\geq sup\{\mathcal{C}_i(\overline{f}(B),r)|(A\leftarrow_i B.r)$ is a ground instance of a rule in $P\}\geq \mathcal{C}_j(\overline{f}(D),c) =\mathcal{C}_j(\overline{f}(D),P(C))$. Hence, $\overline{f}(A\leftarrow_j D)=\leftarrow_j^\bullet(f(A),\overline{f}(D))\geq^{(*)}\leftarrow_j^\bullet(\mathcal{C}_j(\overline{f}(D),P(C)),\overline{f}(D))\geq^{(**)} P(C)$, 
where (*) holds since $\leftarrow_i^\bullet$ is non-decreasing in the first argument, and (**) follows from (\ref{ti03}). Consequently, $\overline{f}(C)=inf\{\overline{f}(A\leftarrow_j D)|(A\leftarrow_j D)$ is a ground instance of $C\}\geq P(C)$.
\end{proof}

Since the given immediate consequences operator $T_P$ satisfies Theorem \ref{th02} and Theorem \ref{th03}, and the set of Herbrand interpretations of the program $P$ is a complete lattice under the relation $\sqsubseteq$, due to Knaster and Tarski \cite{Tarski55}, the Least Herbrand model of the program $P$ is exactly the least fixpoint of $T_P$ and can be obtained by iterating $T_P$ from the bottom interpretation $\bot$ after
$\omega$ iterations, where $\omega$ is the smallest limit ordinal (apart from 0). Furthermore, since the truth domain $\overline{X}$ and the Herbrand base $B_P$ are finite, the least model of $P$ can be obtained after at most $\mathcal{O}(|P||\overline{X}|)$ steps, where $|A|$ denotes the cardinality of the set $A$. This is an important tool for dealing with recursive programs, for which computations can be infinite.

\subsection{Completeness of the procedural semantics}
The following theorem shows that $T_P^n(\bot)$ in fact builds computed answers for ground atoms.
\begin{theorem} \label{th05}
Let $P$ be a program and $A$ a ground atom. For all $n$, there exists a computation for $P$ and the query $?A$ such that the computed answer is $(T_P^n(\bot)(A);id)$.
\end{theorem}
\noindent
\begin{proof}
Note that since $A$ is ground, the substitutions in all computed answers are always the identity. 

We prove the result by induction on $n$.

Suppose first that $n=0$. Since $T_P^0(\bot)(A)=0$, there is a computation for $P$ and $?A$ in which only Rule 2 is applied with the computed answer $(0;id)$.

Now suppose that the result holds for $n-1$, where $n\geq 1$; we prove that it also holds for $n$. There are two cases:

($i$) $A$ does not unify with any rule head or logical part of facts in $P$. Then, $T_P^n(\bot)(A)=0$, and the computation is the same as the case $n=0$. 

($ii$) Otherwise, since the suprema in the definition of $T_P$ are in fact maxima, there exists either a ground instance $(A.b)$ of a fact in $P$ such that $T_P^n(\bot)(A)=b$ or a ground instance $(A\leftarrow_i B.r)$ of a rule  in $P$ such that $T_P^n(\bot)(A)=\mathcal{C}_i(T_P^{n-1}(\bot)(B),r)$. For the former case, there is a computation for $P$ and $?A$ in which only Rule 1 is applied, and the computed answer is $(b;id)$. For the latter, by the induction hypothesis, for each ground atom $B_j$ in $B$, there exists a computation such that $T_P^{n-1}(\bot)(B_j)$ is the computed truth value for $B_j$. Therefore, the computed truth value of the whole body $B$ is $T_P^{n-1}(\bot)(B)$, calculated from all $T_P^{n-1}(\bot)(B_j)$ along the complexity of $B$ using the truth functions of the connectives. Clearly, there is a computation for $P$ and $?A$ in which the first rule to be applied is Rule 1 carried out on the rule in $P$ which has $(A\leftarrow_i B.r)$ as its ground instance, and the rest is a combination of the computations of each $B_j$ in $B$. It is clear that the computed truth value for $?A$ in this computation is $T_P^n(\bot)(A)$.
\end{proof}
The completeness result for the case of ground queries is shown as follows.

\begin{theorem} \label{th06}
For every correct answer $(x;id)$ of a program $P$ and a ground query $?A$, there exists a computed answer $(r;id)$ for $P$ and $?A$ such that $r\geq x$.
\end{theorem}
\noindent
\begin{proof}
Since $(x;id)$ is a correct answer of $P$ and $?A$, for every model $f$ of $P$, we have $f(A)\geq x$. In particular, 
let $M_P$ be the Least Herbrand model of $P$; $M_P(A)=T_P^w(\bot)(A)\geq x$. Recall that $T_P^w(\bot)(A)=sup\{T_P^n(\bot)(A): n<w\}$. Since $w$ is a finite number, the $sup$ operator is in fact a maximum. Hence, there exists $n<w$ such that $T_P^n(\bot)(A)=T_P^w(\bot)(A)$. By Theorem \ref{th05}, there exists a computation for $P$ and $?A$ such that the computed answer is $(T_P^n(\bot)(A);id)$; thus, the theorem is proved.
\end{proof}
The completeness for the case of non-ground queries can be obtained by employing some extended versions of Mgu lemma and Lifting lemma \cite{Ll87} as follows.

We define several more notions. Consider a computation of length $n$ for a program $P$ and a query $?A$; we call each $G_i, i=0...(n-1)$, in the sequence of the computation an \emph{intermediate query}, and the part of the computation from $G_i$ to $G_n$ an \emph{intermediate computation} of length $n-i$. Thus, a computation is a special intermediate computation with $i=0$.
Similar to \citeN{Ll87}, we define an \emph{unrestricted computation} (an \emph{unrestricted intermediate computation}) as a computation (an intermediate computation) in which the substitutions $\theta_i$ in each step are not necessary to be most general unifiers (mgu), but only required to be unifiers. 

In the following proofs, since it is clear for which program a computed answer is, we may omit the program and state that the computed answer is for the (intermediate) query, or the query has the computed answer. The same convention is applied to (unrestricted) (intermediate) computations and correct answers.

\begin{lemma}[Mgu Lemma] \label{lemma1}
Let $P$ be a program and $G_i$ an intermediate query. Suppose that there is an unrestricted intermediate computation for $P$ and $G_i$. Then, there exists an intermediate computation for $P$ and $G_i$ with the same computed truth value and length such that, if $\theta_{i+1},...,\theta_n$ are the unifiers from the unrestricted intermediate computation, and $\theta_{i+1}',...,\theta_n'$ are the mgu's from the intermediate computation, then there exits a substitution $\gamma$ such that $\theta_{i+1}...\theta_n=\theta_{i+1}'...\theta_n'\gamma$.
\end{lemma}
\begin{proof}
The proof is by induction on the length of the unrestricted intermediate computation. 
Suppose first that the length is 1, i.e., $n=i+1$. Since if either Rule 2 or Rule 5 is applied, the unifier is the identity (an mgu), and Rule 1 and Rule 3 cannot be the last rule to be applied in an unrestricted intermediate  computation, the rule to be applied here is Rule 4. Since Rule 4 is the last rule to be applied in the unrestricted intermediate  computation, it can be shown that the unrestricted intermediate computation is also an unrestricted computation of length 1. This means $i=0$. Suppose that $G_0=(A_m;id)$, where $A_m$ is an atom. Then, there exists a fact $(A.b)$ in $P$ such that $\theta_1$ is a unifier of $A_m$ and $A$, and $b$ is the computed truth value. Assume that $\theta_1'$ is an mgu of $A_m$ and $A$. Then, $\theta_1=\theta_1'\gamma$ for some $\gamma$. Clearly, there is a computation for $P$ and $?A_m$  carried out on the same fact $(A.b)$ with length 1, the computed truth value $b$, and the mgu $\theta_1'$.

Now suppose that the result holds for length $\leq k-1$, where $k\geq 2$; we prove that it also holds for length $k$. Assume that there is an unrestricted intermediate computation for $P$ and $G_i$ of length $k$ with the sequence of unifiers $\theta_{i+1},...,\theta_n$, where $n=i+k$. Consider the transition from $G_i$ to $G_{i+1}$. Since $k\geq 2$, it cannot be an application of Rule 5 and thus is one of the following cases:

($i$) Either Rule 2 or Rule 3 is applied. Then, $\theta_{i+1}=id$. By the induction hypothesis, there exists an intermediate computation for $P$ and $G_{i+1}$ of length $k-1$ with mgu's $\theta_{i+2}',...,\theta_n'$ such that $\theta_{i+2}...\theta_n=\theta_{i+2}'...\theta_n'\gamma$ for some $\gamma$. Thus, there is an intermediate computation for $P$ and $G_{i}$ of length $k$ with mgu's $\theta_{i+1}'=id,\theta_{i+2}',...,\theta_n'$ and $\theta_{i+1}...\theta_n=\theta_{i+1}'...\theta_n'\gamma$.

($ii$) Either Rule 1 or Rule 4 is applied. Hence, $\theta_{i+1}$ is a unifier for the selected atom $A$ in $G_i$ and an atom $A'$, which is either a rule head (if Rule 1 is applied) or a logical part of a fact (if Rule 4 is applied) in $P$. There exists an mgu $\theta_{i+1}'$ for $A$ and $A'$ such that $\theta_{i+1}=\theta_{i+1}'\vartheta$ for some $\vartheta$. Therefore, if we use $\theta_{i+1}'$ instead of $\theta_{i+1}$ in the transition, we will obtain an intermediate query $G_{i+1}'$ such that $G_{i+1}=G_{i+1}'\vartheta$ since $G_{i+1}$ and $G_{i+1}'$ are all obtained from $G_{i}$ by replacing $A$ with the same expression, then applying $\theta_{i+1}$ or $\theta_{i+1}'$, respectively. Now consider the transitions from $G_{i+1}$ to $G_{n-1}$. Since they cannot be an application of Rule 5, there are two possible cases:

($a$) All the transitions use only Rule 2 or Rule 3. Thus, all the unifiers are the identity. If we apply the same rule on the corresponding atom (for the case of Rule 2) or on the corresponding body formula (for the case of Rule 3) for each transition from the intermediate query $G_{i+1}'$, we will obtain a sequence $G_{i+1}',...,G_{n-1}'$, and it can be shown that for all $i+1\leq l\leq n-1$, $G_{l}=G_l'\vartheta$. Since the last transition from $G_{n-1}$ to $G_n$ uses Rule 5, $G_{n-1}$ does not have any predicate symbols, and neither does $G_{n-1}'$. Thus, they are identical. As a result, $G_i$ has an intermediate computation $G_i, G'_{i+1},..., G'_{n-1},G_n$ with mgu's $\theta_{i+1}'$ and the identities.

($b$) There exists the smallest $m$ such that $i+1\leq m\leq n-2$, and the transition from $G_{m}$ to $G_{m+1}$ uses either Rule 1 or Rule 4. Hence, all the transitions from $G_{i+1}$ to $G_{m}$ use only Rule 2 or Rule 3. As above, we can have a sequence $G_{i+1}',...,G_{m}'$ such that for all $i+1\leq l\leq m$, $G_{l}=G_l'\vartheta$. Now we will prove the result for the case that Rule 1 is applied in the transition from $G_{m}$ to $G_{m+1}$, and the case for Rule 4 can be proved similarly. The application of Rule 1 in the transition implies that there exists a rule $(A''\leftarrow_j B.r)$ in $P$ such that $\theta_{m+1}$ is a unifier of the selected atom $A_m$ in $G_m$ and $A''$. Since we utilise the usual Prolog renaming of variables along derivation, we can assume that $\vartheta$ does not act on any variables of $A''$ or $B$. Suppose that $A_m'$ is the corresponding selected atom in $G_m'$, we have $A_m=A_m'\vartheta$. Therefore, $\vartheta\theta_{m+1}$ is a unifier for $A_m'$ and $A''$ since $A_m'\vartheta\theta_{m+1}=A_m\theta_{m+1}=A''\theta_{m+1}=A''\vartheta\theta_{m+1}$. Now applying Rule 1 to $G_m'$ on the selected atom $A_m'$ and the rule $(A''\leftarrow_j B.r)$ with the unifier $\vartheta\theta_{m+1}$, we obtain an intermediate query $G_{m+1}'$. Since $(\mathcal{C}_j(B,r))\theta_{m+1}=(\mathcal{C}_j(B,r))\vartheta\theta_{m+1}$ and $G_m=G_m'\vartheta$, we have $G_{m+1}'=G_{m+1}$. Thus, $G_i$ has an unrestricted intermediate computation with the sequence $G_i, G_{i+1}',..., G_m', G_{m+1},..., G_n$ and the unifiers $\theta_{i+1}',\theta_{i+2},...,\theta_{m},\vartheta\theta_{m+1},\theta_{m+2},...,\theta_n$. By the induction hypothesis, $G_m'$ has an intermediate computation with the sequence $G_m',G_{m+1}',...,G_n'$, the mgu's $\theta_{m+1}',...,\theta_n'$, and the same computed truth value such that $\vartheta\theta_{m+1}\theta_{m+2}...\theta_n=\theta_{m+1}'...\theta_n'\gamma$ for some $\gamma$. Since $\theta_{i+2},...,\theta_{m}$ are the identity, $G_i$ has an intermediate computation with the sequence $G_i, G_{i+1}',...,G_m',G_{m+1}',...,G_n'$ and the mgu's $\theta_{i+1}',\theta_{i+2},...,\theta_{m},\theta_{m+1}'...,\theta_n'$, and we have $\theta_{i+1}...\theta_m\theta_{m+1}\theta_{m+2}...\theta_n=\theta_{i+1}'\theta_{i+2}...\theta_m\vartheta\theta_{m+1}\theta_{m+2}...\theta_n=\theta_{i+1}'\theta_{i+2}...\theta_m\theta_{m+1}'...\theta_n'\gamma$.
\end{proof}

\begin{lemma}[Lifting Lemma] \label{lemma2}
Let $P$ be a program, $?A$ a query, and $\theta$ a substitution. Suppose there exists a computation for $P$ and the query $?A\theta$. Then there exists a computation for $P$ and $?A$ of the same length and the same computed truth value such that, if $\theta_1,...,\theta_n$ are mgu's from the computation for $P$ and $?A\theta$, and $\theta_1',...,\theta_n'$ are mgu's from the computation for $P$ and $?A$, then there exists a substitution $\gamma$ such that $\theta\theta_1...\theta_n=\theta_1'...\theta_n'\gamma$.
\end{lemma}
\begin{proof}
The proof is similar to that in \citeN{Ll87}. Suppose that the computation for $P$ and $?A\theta$ has a sequence $G_0=(A\theta;id),G_1, ...,G_n$. Consider the admissible rule to be applied in the transition from $G_0$ to $G_1$. We will prove the result for the case of Rule 1, and it can be proved similarly for the others. The application of Rule 1 implies that there exists a rule $(A'\leftarrow_j B.r)$ in $P$ such that $\theta_1$ is an mgu of $A\theta$ and $A'$. We assume that $\theta$ does not act on any variables of $A'$ or $B$; thus, $\theta\theta_1$ is a unifier for $A$ and $A'$. Now applying Rule 1 to $G_0'=(A;id)$ on the rule $(A'\leftarrow_j B.r)$ with the unifier $\theta\theta_1$, we have $G_1'=G_1$. Therefore, we obtain an unrestricted computation for $P$ and $?A$, which looks like the given computation for $P$ and $?A\theta$, except that the first intermediate query $G_0'$ is different, and the first unifier is $\theta\theta_1$. Now applying the mgu lemma, we obtain the result.
\end{proof}
We also have a lemma which is an extension of Lemma 8.5 in \citeN{Ll87}.
\begin{lemma} \label{lemma3}
Let $P$ be a program and $?A$ a query. Suppose that $(x;\theta)$ is a correct answer for $P$ and $?A$. Then there exists a computation for $P$ and the query $?A\theta$ with a computed answer $(r;id)$ such that $r\geq x$.
\end{lemma}
\begin{proof}
The proof is similar to that in \citeN{Ll87}. Suppose that $A\theta$ has variables $x_1,...,x_n$. Let $a_1,...,a_n$ be distinct constants not appearing in $P$ or $A$, and let $\theta_1$ be the substitution $\{x_1/a_1,...,x_n/a_n\}$. Since for any model $f$ of $P$, $\overline{f}(A\theta\theta_1)\geq\overline{f}(A\theta)\geq x$, and $A\theta\theta_1$ is ground, $(x;id)$ is a correct answer for $P$ and $?A\theta\theta_1$. By Theorem \ref{th06}, there exists a computation for $P$ and $?A\theta\theta_1$ with a computed answer $(r;id)$ such that $r\geq x$. Since the $a_i$ do not appear in $P$ or $A$, by replacing $a_i$ with $x_i$ $(i=1,...,n)$ in this computation, we obtain a computation for $P$ and $?A\theta$  with the computed answer $(r;id)$.
\end{proof}
The completeness of the procedural semantics is stated as follows.

\begin{theorem}
Let $P$ be a program, and $?A$ a query. For every correct answer $(x;\theta)$ for $P$ and $?A$, there exists a computed answer $(r;\sigma)$ for $P$ and $?A$, and a substitution $\gamma$ such that $r\geq x$ and $\theta=\sigma\gamma$.
\end{theorem}
\begin{proof}
Since $(x;\theta)$ is a correct answer for $P$ and $?A$, by Lemma \ref{lemma3}, there exists a computation for $P$ and the query $?A\theta$ with a computed answer $(r;id)$ such that $r\geq x$. Suppose the sequence of mgu's in the computation is $\theta_1,...,\theta_n$. Then $A\theta\theta_1...\theta_n=A\theta$. By the lifting lemma, there exists a computation for $P$ and $?A$ with the same computed truth value $r$ and mgu's $\theta_1',...,\theta_n'$ such that $\theta\theta_1...\theta_n=\theta_1'...\theta_n'\gamma'$, for some substitution $\gamma'$. Let $\sigma$ be $\theta_1'...\theta_n'$ restricted to the variables in $A$. Then $\theta=\sigma\gamma$, where $\gamma$ is an appropriate restriction of $\gamma'$.
\end{proof}
Clearly, the proofs of Mgu and Lifting lemmas here can be similarly applied to fuzzy logic programming and the frameworks of logic programming developed based on it such as \emph{multi-adjoint logic programming} (see, e.g., \citeN{Medina04}).

\subsection{More examples}
\begin{example} \label{ex5}
Assume that we use the truth domain from the 2-limited HA in Example \ref{ex99}, that is, $\underline{X} = (X,\{False,True\},\{V, M, P,L\},\leq)$, and have the following knowledge base:

$(i)$ The sentence ``\emph{A hotel is convenient for a business trip if it is \textbf{very} near to the business location, has a reasonable cost at the time, and is a fine building}" is \emph{Very True}. 

$(ii)$ The sentence ``\emph{A hotel has a reasonable cost if either its dinner cost or its hotel rate at the time is reasonable}" is \emph{Very True}.

$(iii)$ The sentence ``\emph{Causeway hotel is near Midtown Plaza}" is \emph{Little More True}.

$(iv)$ The sentence ``\emph{Causeway hotel is a fine building}" is \emph{Probably More True}.

$(v)$ The sentence ``\emph{Causeway hotel has a reasonable dinner cost in November}" is \emph{Very More True}.

$(vi)$ The sentence ``\emph{Causeway hotel has a reasonable hotel rate in November}" is \emph{Little Probably True}.

Let \emph{cn\_ht, ne\_to, re\_co, fn\_bd, re\_di, re\_rt, Bu\_lo, mt, cw} and \emph{T} stand for ``convenient hotel", ``near to", ``reasonable cost", ``fine building", ``reasonable dinner cost", ``reasonable hotel rate", ``business location", ``Midtown Plaza", ``Causeway hotel", and ``True", respectively. Then, the knowledge base can be represented by the following program:
\begin{eqnarray*}
(cn\_ht(Bu\_lo,Time, Hotel)\leftarrow_G \\
\wedge(V\;ne\_to(Bu\_lo,Hotel), 
re\_co(Hotel,Time), 
fn\_bd(Hotel)).VT)\\
(re\_co(Hotel,Time)\leftarrow_L
\vee(re\_di(Hotel,Time), re\_rt(Hotel,Time)).VT)\\
(ne\_to(mt,cw).LMT)\\
(fn\_bd(cw).PMT)\\
(re\_di(cw,nov).VMT)\\
(re\_rt(cw,nov).LPT)
\end{eqnarray*}
Note that although the conjunctions and disjunction are binary connectives, they can be easily extended to have any arity greater than 2.

Given a query $?cn\_ht(mt,nov,cw)$, we can have the following computation (the substitution in the computed answer is the identity):
\begin{eqnarray*}
?cn\_ht(mt,nov,cw)\\
\mathcal{C}_G(\wedge(V\;ne\_to(mt,cw), re\_co(cw,nov), fn\_bd(cw)),VT)\\
\mathcal{C}_G(\wedge(V^{-}(ne\_to(mt,cw)), re\_co(cw,nov), fn\_bd(cw)),VT)\\
\mathcal{C}_G(\wedge(V^{-}(LMT), re\_co(cw,nov), fn\_bd(cw)),VT)\\
\mathcal{C}_G(\wedge(V^{-}(LMT), re\_co(cw,nov), PMT),VT)\\
\mathcal{C}_G(\wedge(V^{-}(LMT),\mathcal{C}_L(\vee(re\_di(cw,nov), re\_rt(cw,nov)),VT), PMT),VT)\\
\mathcal{C}_G(\wedge(V^{-}(LMT),\mathcal{C}_L(\vee(VMT, re\_rt(cw,nov)),VT), PMT),VT)\\
\mathcal{C}_G(\wedge(V^{-}(LMT),\mathcal{C}_L(\vee(VMT,LPT),VT), PMT),VT)\\
\mathcal{C}_G(\wedge^{\bullet}(V^{-}(LMT),\mathcal{C}_L(\vee^{\bullet}(VMT,LPT),VT), PMT),VT)
\end{eqnarray*}
Using the inverse mappings of hedges in Table \ref{tab3}, we have $\mathcal{C}_G(\wedge^{\bullet}(V^{-}(LMT), \mathcal{C}_L(\vee^{\bullet}($ $VMT,LPT),VT), PMT),VT)=\mathcal{C}_G(\wedge^{\bullet}(V^{-}(LMT), \mathcal{C}_L(VMT,VT), PMT),VT)=\mathcal{C}_G(\wedge^{\bullet}(LPT,PMT,PMT),VT) = \mathcal{C}_G(LPT,VT)=LPT$. Thus, the computed answer is $(LPT;id)$, and the sentence ``\emph{Causeway hotel is convenient for a business trip to Midtown Plaza in November}" is at least \emph{Little Probably True}. 

Now, if we want to relax the first condition in the sentence $(i)$, we can replace the phrase ``\emph{very near to}" by a phrase ``\emph{probably near to}". Then, similarly, we can have a similar program and the following computation:
\begin{eqnarray*}
?cn\_ht(mt,nov,cw)\\
\mathcal{C}_G(\wedge(P\;ne\_to(mt,cw), re\_co(cw,nov), fn\_bd(cw)),VT)\\
...\\
\mathcal{C}_G(\wedge^{\bullet}(P^{-}(LMT),\mathcal{C}_L(\vee^{\bullet}(VMT,LPT),VT), PMT),VT)
\end{eqnarray*}
Using the inverse mappings in Table \ref{tab3}, we have a computed answer $(PMT;id)$.

Similarly, if we remove the hedge for the first condition in the sentence $(i)$, we can have a similar program and the following computation:
\begin{eqnarray*}
?cn\_ht(mt,nov,cw)\\
\mathcal{C}_G(\wedge(ne\_to(mt,cw), re\_co(cw,nov), fn\_bd(cw)),VT)\\
...\\
\mathcal{C}_G(\wedge^{\bullet}(LMT,\mathcal{C}_L(\vee^{\bullet}(MPT,LPT),VT), PMT),VT)
\end{eqnarray*}
Thus, we have a computed answer $(LMT;id)$.

It can be seen that with the same hotel (\emph{Causeway}), the time (\emph{November}), and the business location (\emph{Midtown Plaza}), by similar computations, if we put a higher requirement for the condition ``\emph{near to}", we obtain a lower truth value. More precisely, with the conditions ``\emph{very near to}", ``\emph{near to}", and ``\emph{probably near to}", we obtain the truth values $LPT$, $LMT$, and $PMT$, respectively, and $LPT<LMT<PMT$. This is reasonable and in accordance with common sense.
\end{example}
\section{Applications}
\subsection{A data model for fuzzy linguistic databases with flexible querying}
Information stored in databases is not always precise. Basically, two important issues in research in this field are representation of uncertain information in a database and provision of more flexibility in the information retrieval process, notably via inclusion of linguistic terms in queries. 
Also, the relationship between deductive databases and logic programming has been well established. Therefore, fuzzy linguistic logic programming (FLLP) can provide a tool for constructing fuzzy linguistic databases equipped with flexible querying.

The model is an extension of Datalog \cite{Ul88} without negation and possibly with recursion, which is similar to that in \citeN{PV01}, called \emph{fuzzy linguistic Datalog} (FLDL).
It allows one to find answers to queries over a \emph{fuzzy linguistic database} (FLDB) using a \emph{fuzzy linguistic knowledge base} (FLKB). 
An FLDB is a (crisp) relational database in which an additional attribute is added to every relation to
store a linguistic truth value for each tuple, and an FLKB is a \emph{fuzzy linguistic Datalog program} (FLDL program).
Here, we also work on \emph{safe} rules, i.e., every variable appearing in the rule head of a rule also appears in the rule body.
An FLDL program consists of finite safe rules and facts.
Moreover, in an FLDL program, a fuzzy predicate is either an \emph{extensional database} (EDB) predicate, the logical part of a fact, whose relation is stored in the database, or an \emph{intensional database} (IDB) predicate which is defined by rules, but not both.

We can extend the \emph{monotone subset}, consisting of selection, Cartesian product, equijoin, projection, and union, of relational algebra \cite{Ul88} for the case of our relations and create a new one called \emph{hedge modification}. We call this collection of operations \emph{fuzzy linguistic relational algebra} (FLRA).

Based on the operations, we can convert rules with the same IDB predicate in their heads to an expression of FLRA; the expression yields a relation for the predicate.
Furthermore, it can be observed that the way the expression calculates the truth value of a tuple in the relation for the IDB predicate is the same as the way the immediate consequences operator $T_P$ does for the corresponding ground atom \cite{PV01}.
Thus, similar to the classical case, the FLRA augmented by the immediate consequences operator is sufficient to evaluate recursive FLDL programs, and every query over an FLKB represented by an FLDL program can be exactly evaluated by finitely iterating the operations of FLRA from a set of relations for the EDB predicates.

\subsection{Threshold computation}
This is the case when one is interested in looking for a computed answer to a query with a truth value not less than some threshold $t$. 

Assume that at a certain point in a computation we need to find an answer to the selected atom $A_m$ with a threshold $t_m$. Since $\mathcal{C}_{c}(x,y)\leq min(x,y)$, for $c\in \{L,G\}$, the selected rule or fact which will be used in the next step must have a truth value not less than $t_m$. If there is no such rule or fact, we can cut the computation branch.
For the case that $A_m$ will be unified with the rule head of such a rule, the truth value of the whole body of the rule must not be less than $t_{m+1}=inf\{b|\mathcal{C}(b,r)\geq t_m\}$, where $r$ is the truth value of the rule and $r\geq t_m$. If the implication used in the rule is the G\"{o}del implication, then $t_{m+1}=t_m$; if it is the \L ukasiewicz implication, then $t_{m+1}=v_{n+k-j}$, where $r=v_j, t_m=v_k$ are two values in the truth domain $\overline{X}$, and $v_n=1$. Since $n\geq j\geq k$, we have $t_{m+1}\geq t_m$, and if $r<1$, we have $t_{m+1}>t_m$.

Recall that a rule body can be built from its components using the conjunctions, the disjunction, and hedge connectives. Therefore, we have:
($i$) For the case of G\"{o}del conjunction, $t_{m+1}$ is the next threshold for each of its components, and if $t_{m+1}>t_m$, for all $m$ (this will happen if all the implications are \L ukasiewicz, and all the truth values of rules are less than 1), we can estimate the depth of the search tree according to the threshold $t$ and the highest truth value of rules. 
($ii$) For the case of \L ukasiewicz conjunction, if all the truth values of the facts in the program are less than 1 (thus the computed truth value of any component in any body formula is less than 1), the next threshold for each of the components is greater than $t_{m+1}$. Hence, similar to the above case, we can also work out the depth of the search tree.
($iii$) For the case of disjunction, one of the components of the rule body must have a computed truth value at least $t_{m+1}$.
($iv$) Finally, the problem of finding a computed truth value for a hedge-modified formula $hB$ with a threshold $u$ can be reduced to that of $B$ with a new threshold $u'=inf\{v|h^{-}(v)\geq u\}$.
\subsection{Fuzzy control}
Control theory is aimed at determining a function $\underline{f} :X\rightarrow Y$ whose intended meaning is that given an input value $x$, $\underline{f}(x)$ is the correct value of control signal. A fuzzy approach to control employs an approximation of such a (ideal) function by a system of fuzzy IF-THEN rules of the form ``IF $x$ is $A$ THEN $y$ is $B$'', where $A$ and $B$ are labels of fuzzy subsets.

In the literature, there are several attempts to reduce fuzzy control to fuzzy logic in narrow sense. \citeANP{Gerla05} \citeNN{Gerla05,Gerla01} proposed an interesting reduction in which a fuzzy IF-THEN rule ``IF $x$ is $A$ THEN $y$ is $B$'' is translated into a fuzzy logic programming rule $(good(x,y)\leftarrow A(x)\wedge B(y).\lambda)$, where $A$ and $B$ are now considered as fuzzy predicates. The truth value $\lambda$ is understood as the \emph{degree of confidence} of the experts in such a rule, and by default, $\lambda =1$.
The intended meaning of the new predicate $good(x, y)$ is that given an input value $x$, $y$ is a \emph{good} value for the control variable. 
Therefore, the information carried by a system of fuzzy IF-THEN rules can be represented by a fuzzy logic program.

More precisely, a system of fuzzy IF-THEN rules:
\begin{eqnarray} 
\mbox{IF }x\mbox{ is }A_1\mbox{ THEN }y\mbox{ is }B_1 \nonumber\\
... \label{sys01}\\
\mbox{IF }x\mbox{ is }A_n\mbox{ THEN }y\mbox{ is }B_n \nonumber
\end{eqnarray}
can be associated with the following program $P$:
\begin{eqnarray} 
(good(x,y)\leftarrow A_1(x)\wedge B_1(y).1) \nonumber\\
... \nonumber\\
(good(x,y)\leftarrow A_n(x)\wedge B_n(y).1) \label{sys02} \\
(A_i(r).r_{A_i}) \mbox{, for } r\in X, i=1...n \nonumber\\
(B_j(t).t_{B_j}) \mbox{, for } t\in Y, j=1...n \nonumber
\end{eqnarray}
where $r_{A_i}$ is the degree of truth to which an input value $r$ satisfies a predicate $A_i$, and $t_{B_j}$ is the degree of truth to which an output value $t$ satisfies a predicate $B_j$. 
Each element $r\in X$ or $t\in Y$ is considered as a constant. Thus, the language of $P$ is a two-sorted predicate one, and we have two Herbrand universes $U^{X}_P=X$ and $U^{Y}_P=Y$. Since the truth values of the rules are all equal to 1, \L ukasiewicz and G\"{o}del t-norms yield the same results in computations; therefore, without loss of generality we can use the same notation for the implications.

By iterating the $T_P$ operator from the bottom interpretation $\bot$, we obtain the Least Herbrand model $M_P$ of $P$. In fact, it can be shown that $M_P=T^{2}_P(\bot)$. Let us put $\mathcal{G}(r,t)=M_P(good(r,t))$. Indeed, $\mathcal{G}(r,t)$ can be interpreted as the degree of preference on the output value $t\in Y$, given the input value $r\in X$. 
Therefore, the purpose of the program $P$ is not to compute the ideal function $\underline{f}: X\rightarrow Y$, but to define a fuzzy predicate \emph{good} expressing a graded opinion on a possible control value $t$ w.r.t. a given input value $r$.
Clearly, given an input value $r$, it should be better to take a value $t$ that maximises $\mathcal{G}(r,t)$.
Note that the value $\mathcal{G}(r,t)$ is not a true value, but a lower bound to the truth value of $good(r,t)$. In other words, we can say that given $r$, $t$ can be proved to be good at least at the level $\mathcal{G}(r,t)$.

It is worth noticing that in fuzzy control, it is quite often that the labels of fuzzy subsets in a system of fuzzy IF-THEN rules, i.e., $A_i$ and $B_i$ in the system (\ref{sys01}), are hedge-modified ones, e.g., \emph{Verylarge} and \emph{Veryfast}. Thus, our language can be used to represent the associated program in a very natural way since we allow using linguistic hedges to modify fuzzy predicates. Clearly, in such a program, all the facts $(A_i(r).r_{A_i})$ and $(B_j(t).t_{B_j})$ we need are only for primary predicates (predicates without hedge modification) such as \emph{large} or \emph{fast}, but not for all predicates as in the case of fuzzy logic programming.
\section{Implementation}
In the literature, there has been research on \emph{multi-adjoint logic programming} (MALP) (see, e.g., \citeN{Medina04}), which is an extension of fuzzy logic programming in which truth values can be elements of any complete bounded lattice instead of the unit interval. Also, there have been several attempts to implement systems where multi-adjoint logic programs can be executed. Due to the similarity between MALP and FLLP, the implementation of a system for executing fuzzy linguistic logic programs can be carried out based on the systems built for multi-adjoint ones. 
In the sequel, we sketch an idea for implementing such a system, which is inspired by the FLOPER (Fuzzy LOgic Programming Environment for Research) system described in \citeN{MM08}.

The main objective is to translate fuzzy linguistic logic programs into Prolog ones which can be safely executed inside any standard Prolog interpreter in a completely transparent way. 
We take the following program as an illustrative example:
\begin{eqnarray*}
(gd\_em(X)\leftarrow_G \wedge_L(V\;st\_hd(X),P\;hira\_un(X)).VMT)\\
(hira\_un(ann).VT)\\
(st\_hd(ann).MT)
\end{eqnarray*}
For simplicity, instead of computing with the truth values, we can compute with their indexes in the truth domain. Thus, the program can be coded as:
\begin{eqnarray*}
gd\_em(X)\;<godel\;\; \&luka(hedge\_v(st\_hd(X)),hedge\_p(hira\_un(X))) \;with\; 38.\\
hira\_un(ann)\; with\; 41.\\
st\_hd(ann)\; with\; 36.
\end{eqnarray*}
where 38, 41, and 36 are respectively the indexes of the truth values $VMT$, $VT$, and $MT$ in the truth domain in Example \ref{ex99}.

During the parsing process, the system produces Prolog code as follows:

$(i)$ Each atom appearing in a fuzzy rule is translated into a Prolog atom extended by an additional argument, a truth variable of the form $\_TV_i$, which is intended to store the truth value obtained in the subsequent evaluation of the atom.

$(ii)$ The truth functions of the binary connectives and the t-norms can be easily defined by standard Prolog clauses as follows:
\begin{eqnarray*} 
and\_godel(X,Y,Z)\; :-\; (X=<Y,Z=X;X>Y,Z=Y). \\
and\_luka(X,Y,Z)\; :-\; H\; is\; X+Y-n,(H=<0,Z=0;H>0,Z=H). \\
or\_godel(X,Y,Z)\; :-\; (X=<Y,Z=Y;X>Y,Z=X). 
\end{eqnarray*} 
where $n$ is the index of the truth value 1 in the truth domain (in Example \ref{ex99}, $n= 44$).
Note that $and\_godel$ is the t-norm $\mathcal{C}_G$ as well as the truth function of the conjunction $\wedge$ ($\wedge_G$) while 
$and\_luka$ is the t-norm $\mathcal{C}_L$ and also the truth function of the conjunction $\wedge_L$, and $or\_godel$ is the truth function of the disjunction $\vee$.

Inverse mappings of hedges can be defined by listing all cases in the form of ground Prolog facts (except inverse mappings of 0, $W$, and 1). More precisely, the inverse mappings in Table \ref{tab3} can be defined as follows:
\begin{eqnarray*} 
inv\_map(H,0,0).\\
...\\
inv\_map(l,17,21).\\
...\\
inv\_map(v,33,25).\\
...\\
inv\_map(H,44,44).
\end{eqnarray*} 
where 33, 25, 17, and 21 are indexes of the values $c^{+}$, $Lc^{+}$, $LLc^{-}$, and $VLc^{-}$, respectively; the fact $inv\_map(v,33,25).$ defines the case $V^{-}(c^{+})=Lc^{+}$ while the fact $inv\_map(l,17,21).$ defines the case $L^{-}(LLc^{-})=VLc^{-}$. The facts $inv\_map(H,0,0).$, $inv\_map(H,22,22).$, and $inv\_map(H,44,44).$, where $H$ is a variable of hedges, define the mappings: for all $h$, $h^{-}(0)=0$, $h^{-}(W)=W$, and $h^{-}(1)=1$.

$(iii)$ Each fuzzy rule is translated into a Prolog clause in which the calls to the atoms appearing in its body must be in an appropriate order. More precisely, the call to the atom corresponding to an operation must be after the calls to the atoms corresponding to its arguments in order for the truth variables to be correctly instantiated, and the last call must be to the atom corresponding to the t-norm evaluating the rule. For example, the rule in the previous program can be translated into the following Prolog clause:
\begin{eqnarray*} 
gd\_em(X,\_TV0)\; :-\; st\_hd(X,\_TV1), inv\_map(v,\_TV1,\_TV2), \\
hira\_un(X,\_TV3), inv\_map(p,\_TV3,\_TV4),\\ 
and\_luka(\_TV2,\_TV4,\_TV5), and\_godel(\_TV5,38,\_TV0).
\end{eqnarray*} 
$(iv)$ Each fuzzy fact is translated into a Prolog fact in which the additional argument is just its truth value instead of a truth variable.
For the above program, the two fuzzy facts are translated into two Prolog facts $hira\_un(ann,41)$ and $st\_hd(ann,36)$.

$(v)$ A query is translated into a Prolog goal that is an atom with an additional argument, a truth variable to store the computed truth value. For instance, the query $?gd\_em(X)$ is translated into the Prolog goal: $?-\;gd\_em(X,Truth\_value)$. Given the above program and the above query, a Prolog interpreter will return a computed answer $[X=ann,Truth\_value=29]$, i.e., we have $(gd\_em(ann).PPT)$.

\section{Conclusions and future work}
We have presented fuzzy linguistic logic programming as a result of integrating fuzzy logic programming and hedge algebras. The main aim of this work is to facilitate the representation and reasoning on knowledge expressed in natural languages, where vague sentences are often assessed by a degree of truth expressed in linguistic terms rather than in numbers, and linguistic hedges are usually used to indicate different levels of emphasis. 
It is well known that in order for a formalism to model such knowledge, it should address the twofold usage of linguistic hedges, i.e., in generating linguistic values and in modifying predicates. 
Hence, in this work we use linguistic truth values and allow linguistic hedges as predicate modifiers.
More precisely, in a fuzzy linguistic logic program, each fact or rule  is graded to a certain degree specified by a value in a linguistic truth domain taken from a hedge algebra of a truth variable, and hedges can be used as unary connectives in body formulae. 

Besides the declarative semantics, a sound and complete procedural semantics which directly manipulates linguistic terms is provided to compute a lower bound to the truth value of a query. Thus, it can be regarded as a method of computing with words. 
A fixpoint semantics of logic programs is defined and provides an important tool to handle recursive programs, for which computations can be infinite.

It has been shown that knowledge bases expressed in natural languages can be represented by our language, and
the theory has several applications such as a data model for fuzzy linguistic databases with flexible querying, threshold computation, and fuzzy control. 

Finding more applications for the theory and implementing a system where fuzzy linguistic logic programs can be executed
are directions for our future work.


\begin{thebibliography}{}

\bibitem[\protect\citeauthoryear{Dinh-Khac, H\"{o}lldobler, and Tran}{Dinh-Khac
  et~al\mbox{.}}{2006}]{DK06}
{\sc Dinh-Khac, D.}, {\sc H\"{o}lldobler, S.}, {\sc and} {\sc Tran, D.~K.}
  2006.
\newblock The fuzzy linguistic description logic ${ALC}_{FL}$.
\newblock In {\em Proc. of the 11th International Conference on Information
  Processing and Management of Uncertainty in Knowledge-Based Systems
  ({IPMU}'2006)}. 2096--2103.

\bibitem[\protect\citeauthoryear{Gerla}{Gerla}{2001}]{Gerla01}
{\sc Gerla, G.} 2001.
\newblock Fuzzy control as fuzzy deduction system.
\newblock {\em Fuzzy Sets and Systems\/}~{\em 121}, 409--425.

\bibitem[\protect\citeauthoryear{Gerla}{Gerla}{2005}]{Gerla05}
{\sc Gerla, G.} 2005.
\newblock Fuzzy logic programming and fuzzy control.
\newblock {\em Studia Logica\/}~{\em 79}, 231--254.

\bibitem[\protect\citeauthoryear{H\'{a}jek}{H\'{a}jek}{1998}]{Ha98}
{\sc H\'{a}jek, P.} 1998.
\newblock {\em Metamathematics of Fuzzy Logic}.
\newblock Kluwer, Dordrecht, The Netherlands.

\bibitem[\protect\citeauthoryear{Kraj\v{c}i, Lencses, and
  Vojt\'{a}\v{s}}{Kraj\v{c}i et~al\mbox{.}}{2004}]{kra04}
{\sc Kraj\v{c}i, S.}, {\sc Lencses, R.}, {\sc and} {\sc Vojt\'{a}\v{s}, P.}
  2004.
\newblock A comparison of fuzzy and annotated logic programming.
\newblock {\em Fuzzy Sets and Systems\/}~{\em 144}, 173--192.

\bibitem[\protect\citeauthoryear{Lloyd}{Lloyd}{1987}]{Ll87}
{\sc Lloyd, J.~W.} 1987.
\newblock {\em Foundations of logic programming}.
\newblock Springer Verlag, Berlin, Germany.

\bibitem[\protect\citeauthoryear{Medina, Ojeda-Aciego, and
  Vojt\'{a}\v{s}}{Medina et~al\mbox{.}}{2004}]{Medina04}
{\sc Medina, J.}, {\sc Ojeda-Aciego, M.}, {\sc and} {\sc Vojt\'{a}\v{s}, P.}
  2004.
\newblock Similarity-based unification: a multi-adjoint approach.
\newblock {\em Fuzzy Sets and Systems\/}~{\em 146,\/}~1, 43--62.

\bibitem[\protect\citeauthoryear{Morcillo and Moreno}{Morcillo and
  Moreno}{2008}]{MM08}
{\sc Morcillo, P.~J.} {\sc and} {\sc Moreno, G.} 2008.
\newblock Using floper for running/debugging fuzzy logic programs.
\newblock In {\em Proc. of the 12th International Conference on Information
  Processing and Management of Uncertainty in Knowledge-Based Systems
  (IPMU'2008)}, {L.~Magdalena}, {M.~Ojeda-Aciego}, {and} {J.~Verdegay}, Eds.
  M\'alaga, 481--488.

\bibitem[\protect\citeauthoryear{Naito, Ozawa, Hayashi, and Wakami}{Naito
  et~al\mbox{.}}{1995}]{NOHW95}
{\sc Naito, E.}, {\sc Ozawa, J.}, {\sc Hayashi, I.}, {\sc and} {\sc Wakami, N.}
  1995.
\newblock A proposal of a fuzzy connective with learning function and query
  networks for fuzzy retrieval systems.
\newblock In {\em Fuzziness in Database Management Systems}, {P.~Bosc} {and}
  {J.~Kacprzyk}, Eds. Physica-Verlag, 345--364.

\bibitem[\protect\citeauthoryear{Nguyen, Tran, Huynh, and Nguyen}{Nguyen
  et~al\mbox{.}}{1999}]{HoKhang99}
{\sc Nguyen, C.~H.}, {\sc Tran, D.~K.}, {\sc Huynh, V.~N.}, {\sc and} {\sc
  Nguyen, H.~C.} 1999.
\newblock Linguistic-valued logic and their application to fuzzy reasoning.
\newblock {\em International Journal of Uncertainty, Fuzziness and
  Knowledge-Based Systems\/}~{\em 7}, 347--361.

\bibitem[\protect\citeauthoryear{Nguyen, Vu, and Le}{Nguyen
  et~al\mbox{.}}{2008}]{Ho08}
{\sc Nguyen, C.~H.}, {\sc Vu, N.~L.}, {\sc and} {\sc Le, X.~V.} 2008.
\newblock Optimal hedge-algebras-based controller: Design and application.
\newblock {\em Fuzzy Sets and Systems\/}~{\em 159}, 968--989.

\bibitem[\protect\citeauthoryear{Nguyen and Wechler}{Nguyen and
  Wechler}{1990}]{Ho90}
{\sc Nguyen, C.~H.} {\sc and} {\sc Wechler, W.} 1990.
\newblock Hedge algebras: An algebraic approach to structure of sets of
  linguistic truth values.
\newblock {\em Fuzzy Sets and Systems\/}~{\em 35}, 281--293.

\bibitem[\protect\citeauthoryear{Nguyen and Wechler}{Nguyen and
  Wechler}{1992}]{HoWech92}
{\sc Nguyen, C.~H.} {\sc and} {\sc Wechler, W.} 1992.
\newblock Extended hedge algebras and their application to fuzzy logic.
\newblock {\em Fuzzy Sets and Systems\/}~{\em 52}, 259--281.

\bibitem[\protect\citeauthoryear{Pokorn\'{y} and Vojt\'{a}\v{s}}{Pokorn\'{y}
  and Vojt\'{a}\v{s}}{2001}]{PV01}
{\sc Pokorn\'{y}, J.} {\sc and} {\sc Vojt\'{a}\v{s}, P.} 2001.
\newblock A data model for flexible querying.
\newblock In {\em Advances in Databases and Information Systems, ADBIS''01},
  {A.~Caplinskas} {and} {J.~Eder}, Eds. Springer Verlag, 280–--293.
\newblock LNCS 2151.

\bibitem[\protect\citeauthoryear{Tarski}{Tarski}{1955}]{Tarski55}
{\sc Tarski, A.} 1955.
\newblock A lattice-theoretical fixpoint theorem and its applications.
\newblock {\em Pacific Journal of Mathematics\/}~{\em 5}, 285--309.

\bibitem[\protect\citeauthoryear{Ullman}{Ullman}{1988}]{Ul88}
{\sc Ullman, J.~D.} 1988.
\newblock {\em Principles of Database and Knowledge-Base Systems}. Vol.~I.
\newblock Computer Science Press, United States of America.

\bibitem[\protect\citeauthoryear{Vojt\'{a}\v{s}}{Vojt\'{a}\v{s}}{2001}]{Vo01}
{\sc Vojt\'{a}\v{s}, P.} 2001.
\newblock Fuzzy logic programming.
\newblock {\em Fuzzy Sets and Systems\/}~{\em 124}, 361--370.

\bibitem[\protect\citeauthoryear{Zadeh}{Zadeh}{1972}]{Zadeh72}
{\sc Zadeh, L.~A.} 1972.
\newblock A fuzzy-set-theoretic interpretation of linguistic hedges.
\newblock {\em Journal of Cybernetics\/}~{\em 2,\/}~3, 4--34.

\bibitem[\protect\citeauthoryear{Zadeh}{Zadeh}{1975a}]{Zadeh75}
{\sc Zadeh, L.~A.} 1975a.
\newblock The concept of a linguistic variable and its application in
  approximate reasoning.
\newblock {\em Information Sciences\/}~{\em 8, 9}, 199--249, 301--357, 43--80.

\bibitem[\protect\citeauthoryear{Zadeh}{Zadeh}{1975b}]{Zadeh75b}
{\sc Zadeh, L.~A.} 1975b.
\newblock Fuzzy logic and approximate reasoning.
\newblock {\em Synthese\/}~{\em 30}, 407--428.

\bibitem[\protect\citeauthoryear{Zadeh}{Zadeh}{1979}]{Zadeh79}
{\sc Zadeh, L.~A.} 1979.
\newblock A theory of approximate reasoning.
\newblock In {\em Machine Intelligence}, {J.~E. Hayes}, {D.~Michie}, {and}
  {L.~I. Mikulich}, Eds. Vol.~9. Wiley, 149--194.

\bibitem[\protect\citeauthoryear{Zadeh}{Zadeh}{1989}]{Zadeh89}
{\sc Zadeh, L.~A.} 1989.
\newblock Knowledge representation in fuzzy logic.
\newblock {\em IEEE Transactions on Knowledge and Data Engineering\/}~{\em
  1,\/}~1, 89--99.

\bibitem[\protect\citeauthoryear{Zadeh}{Zadeh}{1997}]{Zadeh97}
{\sc Zadeh, L.~A.} 1997.
\newblock Toward a theory of fuzzy information granulation and its centrality
  in human reasoning and fuzzy logic.
\newblock {\em Fuzzy Sets and Systems\/}~{\em 90}, 111--127.

\end{thebibliography}
\end{document}